\documentclass[twocolumn,10pt]{IEEEtran}
\usepackage{algorithm,algorithmic,amsbsy,amsmath,amssymb,epsfig,mathrsfs,fancyhdr,subfigure,url,cite,color}
\renewcommand{\baselinestretch}{0.98}

\newcommand{\beq}{\begin{equation}}
\newcommand{\eeq}{\end{equation}}
\newcommand{\beqn}{\begin{eqnarray}}
\newcommand{\eeqn}{\end{eqnarray}}

\newtheorem{theorem}{Theorem}[section]
\newtheorem{lemma}[theorem]{Lemma}

\newenvironment{proof}[1][Proof]{\begin{trivlist}
\item[\hskip \labelsep {\bfseries #1}]}{\end{trivlist}}

\newenvironment{remark}[1][Remark]{\begin{trivlist}
\item[\hskip \labelsep {\bfseries #1}]}{\end{trivlist}}

\newcommand{\qed}{\nobreak \ifvmode \relax \else
      \ifdim\lastskip<1.5em \hskip-\lastskip
      \hskip1.5em plus0em minus0.5em \fi \nobreak
      \vrule height0.75em width0.5em depth0.25em\fi}

\begin{document}

\title{Virtualization of 5G Cellular Networks as a Hierarchical Combinatorial Auction}
\author{Kun Zhu and Ekram Hossain\thanks{K. Zhu is with College of Computer Science and Technology at Nanjing University of Aeronautics and Astronautics. He was with the Department of Electrical and Computer Engineering at the University of Manitoba. E. Hossain is with the Department of Electrical and Computer Engineering at the University of Manitoba, Canada (emails: kunzhu@hotmail.com, Ekram.Hossain@umanitoba.ca). The work was supported by a Discovery Grant from the Natural Sciences and Engineering Research Council of Canada (NSERC).}}
\maketitle

\begin{abstract}

Virtualization has been seen as one of the main evolution trends in
the forthcoming fifth generation (5G) cellular networks which enables the decoupling of
infrastructure from the services it provides. In this case, the
roles of infrastructure providers (InPs) and mobile virtual network
operators (MVNOs) can be logically separated and the resources
(e.g., subchannels, power, and antennas) of a base station owned by
an InP can be transparently shared by multiple MVNOs, while each
MVNO virtually owns the entire BS. Naturally, the issue of resource
allocation arises. In particular, the InP is required to abstract the
physical resources into isolated slices for each MVNO who then
allocates the resources within the slice to its subscribed users. In
this paper, we aim to address this two-level hierarchical resource
allocation problem while satisfying the requirements of efficient
resource allocation, strict inter-slice isolation, and the ability
of intra-slice customization. To this end, we design a hierarchical
combinatorial auction mechanism, based on which a truthful and
sub-efficient resource allocation framework is provided.
Specifically, winner determination problems (WDPs) are formulated
for the InP and MVNOs, and computationally tractable algorithms are
proposed to solve these WDPs. Also, pricing schemes are designed
to ensure incentive compatibility. The designed mechanism can
achieve social efficiency in each level even if each party involved
acts selfishly. Numerical results show the effectiveness of the
proposed scheme.

\end{abstract}

\begin{IEEEkeywords}
5G cellular, massive MIMO, wireless network virtualization, resource
allocation, mechanism design, hierarchical combinatorial auction,
incentive compatibility, winner determination problem (WDP).
\end{IEEEkeywords}

\section{Introduction}

The next generation cellular wireless networks (i.e., 5G networks)
are expected to be deployed around 2020 which are envisioned to
provide higher data rate, lower end-to-end latency, improved
spectrum/energy efficiency, and reduced cost per bit. In general,
addressing these requirements will require significantly larger
amount of spectrum, more aggressive frequency reuse, extreme
densification of small cells, and the wide use of several enabling
technologies (e.g., full-duplex, massive MIMO, C-RAN, and wireless
virtualization)~\cite{Hossain2015}. In this paper, we will focus on
the issue of wireless virtualization which has been receiving
increasing attentions from both academia and
industry~\cite{Liang2014,Costa2013,NEC2013}.

The main idea of wireless virtualization is to enable resource
sharing and to decouple the infrastructure from the services it
provides. Accordingly, the role of infrastructure provider (InP)
needs to be logically separated from the role of service
providers\footnote{We assume that the InP does not provide services
directly to users.}. And the InP can provide the infrastructure as a
service (IasS) to mobile virtual network operators (MVNOs) who may
not have their own infrastructure and/or wireless network resources.
Specifically, the physical resources (e.g., infrastructure,
spectrum, power, backhaul/fronthaul, and antennas) of a base station
(BS) owned by an InP are abstracted into isolated virtual resources
(i.e., slices) which are then transparently shared among different
MVNOs. Each MVNO virtually owns the entire BS with the resources
provided in the allocated slice.

Several benefits can be achieved through such decoupling and
sharing (i.e., virtualization). First, the resource utilization can be improved through
moderating the dynamic requirements of users from different MVNOs
(i.e., statistical multiplexing gains). Second, the capital
expenses (CapEx) and operation expenses (OpEx) can be reduced
through sharing. Third, lower entry barrier for small service
providers could enrich the services provided to users.

A significant challenge for wireless virtualization is resource
allocation which addresses the problem of how to slice the physical
resources for virtual networks of MVNOs to accommodate the dynamic
demands of their subscribed users, while satisfying the requirements
of efficient resource allocation, inter-slice isolation, and the
ability of intra-slice customization. The problem is more
challenging if the agents involved are self-interested. That is, how
to design a mechanism that can achieve a desirable social
efficiency\footnote{By desirable social efficiency we mean that the
social welfare, that results from the behaviors of all agents, can achieve a
certain required value.} even when each agent acts selfishly.

In general, there are two types of implementation schemes for
resource allocation in wireless virtualization. In the first type,
the InP plays the central role who directly allocates the physical
resources to users of different MVNOs according to certain
requirements (e.g., pre-determined resource sharing ratios). In the
second type, the MVNOs are also involved which makes the resource
allocation problem a hierarchical (i.e., two-level) problem. In this case,
the InP is only responsible for allocating the resources to each
MVNO, while each MVNO manages the resource allocation for its users.

Most of the existing work on resource allocation for wireless
virtualization can be categorized into the first type. Specifically,
in~\cite{Malanchini2014,Belt2014,Kamel2014,Lu2014},
optimization-based dynamic resource allocation schemes were
proposed. In~\cite{Fu2013}, a stochastic game based scheme was
proposed. The schemes in these work can achieve high resource
utilization. However, since the MVNOs are not involved in the
resource allocation, the capability of intra-slice customization for
each MVNO cannot be easily achieved. Besides, the computation
complexity for InP is high considering that the optimal resource
allocation has to be obtained directly for all users. A few work
considered the problem of resource allocation to MVNOs. For example,
in~\cite{Yang2013}, an opportunistic sharing based resource
allocation scheme was proposed. In~\cite{Liu2013}, a bankruptcy game
was proposed for dynamic wireless resource allocation among multiple
operators. However, in these work the users are not involved.
Besides, most of the existing work on wireless virtualization do not
consider techniques such as massive MIMO~\cite{Larsson2014}, which
will be a key enabler for 5G networks. Also, the social efficiency
was not considered when agents play selfishly.

To jointly address the two-level resource allocation problem, we
design a hierarchical auction mechanism consisting of two
hierarchical auction models (i.e., a single-seller multiple-buyer
model as shown in Fig.~\ref{fig:BasicAuctionModel} and an extended
multiple-seller multiple-buyer model as shown in
Fig.~\ref{fig:ExtendedAuctionModel}), allocation procedures, and the
corresponding pricing schemes. Note that auction approaches have
been widely used in the literature for resource allocation problem
in wireless systems~\cite{Zhang2013}. This is due to the efficiency
in both process and outcome. With the advent of the concepts of
wireless network virtualization and spectrum secondary market in
cognitive radio networks, the middlemen (e.g., MVNOs) play more
important roles in bridging the supply from resource owners with the
demands from end users. However, {\em almost all existing work on
application of auction mechanisms for resource allocation only
consider single-level auctions without involving the middlemen}. To
the best of our knowledge, the work in~\cite{Tang2012} is the first
work on this topic in which some qualitative analyses were provided
but without application. Besides, there is a social planner
controlling the entire resource allocation of all players. In this
case, it cannot satisfy the requirement of intra-slice customization
for wireless virtualization. The work in~\cite{Lin2013} proposed a
three-stage auction mechanism for spectrum trading. However, the
middleman can submit bids for multiple items but can acquire at most
one. The work in~\cite{Lin2014} investigated the hierarchical
resource allocation through an auction in the upper level and a
price-demand method in the lower level. Besides, all these work only
consider the single seller case. The main differences of the
proposed scheme with existing work are summarized in Table I in 
\textbf{Appendix A}. Note that compared with hierarchical game
based approach (e.g., Stackelberg game in ~\cite{Duan2010}) in which
the social welfare under equilibrium strategies may not be the
optimal ones in sub-games, the proposed hierarchical auction
mechanism achieves social efficiency at each level\footnote{Note
that there is a gap between the allocation efficiency achieved by a
general sharing scheme (i.e., the InP directly allocates resources
to all users) and that achieved by the hierarchical auction
mechanism. Therefore, we use the term `sub-efficient' for the entire
hierarchical auction. This can also be seen as the cost due to the
introduction of middlemen.}.

Based on the proposed hierarchical auction mechanism, a truthful and
sub-efficient semi-distributed resource allocation framework for
wireless virtualization is provided. Specifically, the proposed
hierarchical auction models consist of two levels of combinatorial
auctions. In the lower-level auction, the users act as the bidder
and each MVNO acts as a seller, while in the upper-level auction,
the MVNOs act as the bidders and the InP acts as the seller. The
role of an MVNO can be regarded as that of a middleman. Note that
the two-level auctions are dependent considering the fact that the
MVNOs do not have intrinsic demands and valuations, which however
depend on the demands and resale gains from users.

To determine the optimal amount of resources allocated to each
bidder, winner determination problems (WDPs) are formulated, and the
corresponding solution algorithms are proposed. Also, pricing
schemes are designed to ensure incentive compatibility, which is
critical for achieving social efficiency. The proposed scheme can
satisfy the three requirements mentioned above. First, efficient
resource allocation can be achieved in each level by allocating
resources to bidders with higher valuations. Second, we can achieve
flexible but strict inter-slice isolation in the sense that once the
resources are allocated to an MVNO, these resources can only be
allocated to the users of that MVNO (i.e., strict), while the amount
of resources allocated to each MVNO can be dynamically changed
(i.e., flexible). Third, intra-slice customization can be achieved
due to the involvement of MVNOs which can individually decide how
the resources within the slice can be allocated. Also, compared with
totally centralized allocation schemes (e.g., schemes of the first
type), the computational complexity can be reduced considering two
facts. First, the computation is distributed among InP and MVNOs,
each of which only needs to calculate the allocation for its own
bidders. Second, the dimension of the winner determination problem
faced by each MVNO is relatively small considering that the
available resources are only a subset of the entire BS resources.

The novelty and main contributions of this paper can be summarized
as follows:
\begin{itemize}
    \item A hierarchical combinatorial auction mechanism is designed to jointly address the hierarchical resource allocation for wireless virtualization in massive MIMO networks.
    \item The proposed scheme can satisfy the three requirements of wireless
    virtualzaiton. Also, it jointly considers the feasibility, admission control, and allocation
    problems in a unified resource allocation framework.
    \item Several desirable properties of the hierarchical auction mechanism (i.e., incentive compatibility, individual rationality, and allocation efficiency in each level) can be achieved with appropriate design of
    allocation and pricing schemes.
    \item The computations are migrated to different parties which
    lowers the complexity for each party. Also, computationally tractable algorithms are proposed for solving the WDPs.
    \item While most of the existing auction-based resource allocation schemes only consider
    one dimensional resource (e.g., in most cases the spectrum), we
    consider the allocation with more degrees of freedom (i.e., frequency, power, and
    spatial). Also, most existing work only consider single-minded bidder,
    while we consider both single-minded and
    general valuation bidders\footnote{The definitions of single-minded and general valuations will be given in Section III.}.
\end{itemize}

The rest of the paper is organized as follows. Section II describes
the system model, assumptions, and presents the proposed
hierarchical combinatorial auction models. In Section III, the
resource allocation for wireless virtualization is investigated. WDP
formulations are given and the corresponding solution algorithms and
pricing schemes are presented. Also, theoretical analysis of the
auction properties is provided. In Section IV, the auction model is
extended to consider multi-seller multi-buyer case. Numerical
results and analysis are presented in Section V. Section VI
concludes the paper.

\section{System Model, Assumptions, and Proposed Hierarchical Combinatorial Auction Models}

\subsection{Channel Model and Assumptions}

For the system model, we consider the downlink transmission of an
OFDMA-based cellular system with an InP providing infrastructure
services (including base stations and wireless resources) to a set
of $\mathcal{M}=\{1,2,\ldots,M\}$ MVNOs. Each MVNO $m$ then provides
services to $K_m$ subscribed users in the considered cell. The InP
owns a set of $\mathcal{C}=\{1,2,\ldots,C\}$ subchannels each with
bandwidth $W$. Universal frequency reuse is considered for all
cells. The base station (BS) in each cell is equipped with $A$
antennas and each user equipment has a single antenna (i.e.,
multi-user MIMO). We assume $A$ to be large (e.g., several
hundreds) to achieve massive MIMO effect which scales up traditional
MIMO by orders of magnitude.

The multi-cell system is operated in time-division duplexing (TDD)
mode and we assume that all BSs and UEs are perfectly synchronized.
For achieving all the benefits of massive MIMO, the base station
requires channel state information for precoding. To this end,
channel reciprocity is exploited. That is the downlink channel is
obtained by the Hermitian transpose of the uplink channel, which can
be estimated by the BS from the uplink pilots transmitted by each
user. Note that by operating in TDD mode, the massive MIMO system is
scalable in the sense that the time required for pilots is
proportional to the number of users served per cell and is
independent of the number of antennas~\cite{Larsson2014}. Also, the
channel responses are assumed to be invariant during the symbol
time.

Massive MIMO uses spatial-division multiplexing. In this case, the
BS can serve multiple users in the same time-frequency resource
block. We assume that for each subchannel, the maximum number of
users can be served simultaneously is $J$, which is limited by the
coherence time and accordingly the number of orthogonal pilots.

The downlink received signal for a user $k$
of MVNO $m$ on subchannel $n$ in the considered cell is given by
\beqn
    \hat{s}_k^n = \sqrt{p_k^n d_k g_k} \mathbf{h}_{k}^\mathrm{T}(n)\mathbf{f}_{k}(n)s_k^n +
    \sum_{j\neq k} \sqrt{p_j^n d_k g_k} \nonumber \\  \mathbf{h}_{k}^\mathrm{T}(n)
    \mathbf{f}_{j}(n)s_j^n +
    \sum_{l\neq o} \sqrt{p_l^n d_{kl} g_{kl}} \mathbf{h}_{kl}^\mathrm{T}(n)
    \mathbf{f}_{kl}(n)s_{l}^n + z_k^n,
\label{eq:ReceSignal} \eeqn where the first term in the right hand
side of (\ref{eq:ReceSignal}) denotes the desired signal for user
$k$, while the second and the third terms represent the subchannel
reuse interference within the cell and the interference from other
cells, respectively. Specifically, $p_k^n$ is the transmit power for
the link from the BS to user $k$ in subchannel $n$, $d_k$ and $g_k$
represent the path-loss and the shadowing gain, respectively.
$\mathbf{h}_{k}(n) \in \mathbb{C}^{A_m \times 1}$ denotes the
small-scale fading between the BS and user $k$ on subchannel $n$ which is
the Hermitian transpose of the uplink channel, where $A_m$ is the number of
antennas allocated to MVNO $m$. $\mathbf{f}_{k}(n) \in
\mathbb{C}^{A_m \times 1}$ represents the precoding vector used by
the BS. $s_k^n$ is the transmit signal, and $z_k^n$ denotes the
additive noise with distribution $\mathcal{CN}(0,N_0)$, where $N_0$
is the noise power spectral density. Accordingly, the received
signal-to-interference-plus-noise ratio (SINR) for user $k$ in
subchannel $n$ can be expressed as \beqn \Gamma_k^n = \frac{p_k^n
d_k g_k\mathbf{f}_{k}^T(n)
\mathbf{h}_{k}(n)\mathbf{h}_{k}^\mathrm{T}(n)
\mathbf{f}_{k}(n)}{WN_0 + I_{reuse} + I_{other cell}}, \eeqn where
$I_{reuse}$ and $I_{other cell}$ represent the interference terms.
The ergodic achievable downlink rate for user $k$ in subchannel
$n$ can be obtained as \beqn
    r_k^n = W\mathbb{E}[\log(1+\Gamma_k^n)].
\eeqn

The above ergodic achievable rate is difficult to calculate for
finite system dimensions. Instead, to achieve a tight approximation for finite
systems, an asymptotic analysis is
performed in~\cite{Hoydis2013} assuming that the number of antennas
$A_m$ and the number of users $K_m$ approach infinity while
keeping a finite ratio. Based on the results in~\cite{Hoydis2013}, ignoring
estimation noise and considering inter-user interference to be
negligible compared with noise and pilot contamination in large
scale multiuser (MU)-MIMO systems (the user channels decorrelate with large
number of BS antennas, and strong desired signal can be received
with little inter-user interference~\cite{Bjornson2014}), a
deterministic approximation of SINR for user $k$ of MVNO $m$ in
subchannel $n$ can be obtained as
\begin{eqnarray}
    \hat{\Gamma}_k(n) = \frac{1}{\frac{\bar{L}}{\rho_k(n) A_m}+\alpha(\bar{L}-1)},
\label{eq:SINR}
\end{eqnarray}
where $\rho_k(n)$
represents the transmit SNR, $\bar{L}=1+\alpha(L-1)$ and $L$ represents the number of
cells, $\alpha$ represents the intercell interference factor,
$\alpha(\bar{L}-1)$ is the pilot contamination caused by the reuse
of pilot sequences in other cells, which primarily limits the performance of massive MIMO systems~\cite{Marzetta2010}. Accordingly, the approximate achievable
downlink rate for user $k$ can be expressed as
\begin{eqnarray}
     r_k = \sum_{n\in \mathcal{C}_m} y_k(n)
     W\log(1+\hat{\Gamma}_k(n)),
\label{eq:DLRate}
\end{eqnarray}
where $\mathcal{C}_m$ is the set of subchannels allocated to MVNO
$m$, $y_k(n)$ is the assignment indicator with $y_k(n)=1$ indicating
subchannel $n$ is assigned to user $k$, and $y_k(n)=0$ otherwise.

Note that equation~(\ref{eq:SINR}) is obtained without assuming
that $A_m \gg K_m$ (i.e., assume $A_m \rightarrow \infty$ while keeping
$K_m$ fixed when analyzing the SINR). Instead, it considers more
practical settings where the number of BS antennas is not extremely
large compared with the number of users. This is suitable for our
system model since the BS antennas need to be partitioned for
different MVNOs.

\subsection{Wireless Virtualization Model}

For wireless network virtualization, isolation among different
virtualized wireless networks for different MVNOs is a basic
requirement which can be done at different levels (e.g., flow
level~\cite{Kokku2012,Kokku2013} and physical resource level
\cite{Malanchini2014,Belt2014,Kamel2014,Lu2014,Yang2013,Liu2013}).
In general, isolation at a higher level is simpler for
implementation while at the cost of possible inefficient allocation
and non-strict isolation. In contrast, isolation at a lower level
could achieve better resource utilization at the cost of higher
computational complexity. In this work, we consider the isolation to
be performed at physical resource  (i.e., subchannel, power,
and antenna) level.

Also, isolation at the physical resource level can be implemented in
different manners. The first is a static fixed sharing scheme with
which the MNVOs are preassigned a fixed subset of physical resources
in different domains, and the access is restricted within this fixed
subset. The second is a dynamic general sharing scheme with which
there is no restriction on the resource access, while the isolation
is achieved by guaranteeing certain pre-determined requirements
(e.g., minimum share of the resources). In this work, we adopt a
hybrid isolation scheme in between. Specifically, the InP reserves
certain amounts of resources for each MVNO according to
pre-determined service agreements, while the leftover resources can
be dynamically shared by all MVNOs (e.g., through auctions). Note
that this model can also be applied to the case that some MVNOs own
certain resources.

\subsection{The Proposed Hierarchical Auction Models}

In general, an auction process involves the following entities: a) bidders
who want to buy certain commodities, b) sellers who want to sell
certain commodities, and c) an auctioneer who hosts and directs the
auction process. An auction mechanism mainly involves the
following procedures:
\begin{itemize}
    \item \emph{Bidding procedure}: each bidder $i$ places a bid $b_i$ according to its own valuation $v_i$ of the
    item/items to be auctioned. The valuation is a private information
    which represents the maximum a bidder is willing to pay for the
    item/items. Different bidders could have different valuations for the
    same item/items.
    \item \emph{Allocation procedure}: after collecting the bids from
    all participating bidders, the auctioneer needs to determine how to allocate the
    item/items among the bidders for achieving certain objectives. A bidder is a winning
    bidder if the resource requirement in her bid is satisfied.
    \item \emph{Pricing procedure}: after determining the allocation, the
    auctioneer also needs to determine the price $q_i$ charged to
    each winning bidder $i$.
\end{itemize}

In this work, two hierarchical (i.e., two-level) combinatorial
auction models are proposed for wireless network virtualization as
shown in Fig.~\ref{fig:BasicAuctionModel} and
Fig.~\ref{fig:ExtendedAuctionModel}.

1) \emph{Single-seller multiple-buyer hierarchical auction model:}
We first consider a single-seller multiple-buyer case (i.e., the
model in Fig.~\ref{fig:BasicAuctionModel}) with which the bidders
(e.g., users) can only acquire resources from a single seller. The
entire hierarchical auction consists of sub-auctions in two
levels. Specifically, in the upper level, the InP, who owns the
physical resources, holds a sub-auction and acts as the seller as
well as the auctioneer. The MVNOs act as the bidders (i.e., buyers).
In the lower level, each MVNO then holds a sub-auction acting as the
seller and the subscribed users act as the bidders (accordingly,
there are $M$ sub-auctions in the lower level).

As stated in Section II-B, we consider a hybrid isolation scheme. In
this case, each MVNO $m$ reserves $C_m^{\mathrm{res}}$ number of
subchannels, $P_m^{\mathrm{res}}$ amount of power, and
$A_m^{\mathrm{res}}$ number of antennas, where $\sum_{m=1}^M
C_m^{\mathrm{res}} \leq C$, $\sum_{m=1}^M P_m^{\mathrm{res}} \leq
P$, and $\sum_{m=1}^M A_m^{\mathrm{res}} \leq A$. The leftover
$C^{\mathrm{up}} = C-\sum_{m=1}^M C_m^{\mathrm{res}}$ number of
subchannels, $P^{\mathrm{up}} = P-\sum_{m=1}^M P_m^{\mathrm{res}}$
amount of power, and $A^{\mathrm{up}} = A-\sum_{m=1}^M
A_m^{\mathrm{res}}$ number of antennas are the commodities to be
auctioned among the MVNOs in the upper level auction. Denote by $C_m
\geq 0$, $P_m \geq 0$, and $A_m \geq 0$ the resources obtained by
MVNO $m$ in the upper-level auction, then for each user of MVNO $m$
in the lower-level auction, the available resources are $\hat{C}_m =
C_m^{\mathrm{res}} + C_m$, $\hat{P}_m = P_m^{\mathrm{res}} + P_m$,
and $\hat{A}_m = A_m^{\mathrm{res}} + A_m$.

Note that the hierarchical auction model is not a simple combination
of two levels of separated auctions due to the involvement of
middlemen. Specifically, unlike the users (i.e., bidders) in the
lower-level auction, the MVNOs as middlemen do not have intrinsic
demands\footnote{Note that the demand of an MVNO is not simply the
sum of the demands from all users, since it may not be optimal for
her.}. Furthermore, the MVNOs do not have intrinsic valuations of
the resources, and the valuation depends on the resale revenue which
is also shown in~\cite{Tang2012}. In this case, the two level
auctions are interrelated and should be studied jointly (e.g., in a
way similar to that for analyzing a hierarchical game).

\begin{figure}[th]
    \begin{center}
    \includegraphics[width=0.25\textwidth]{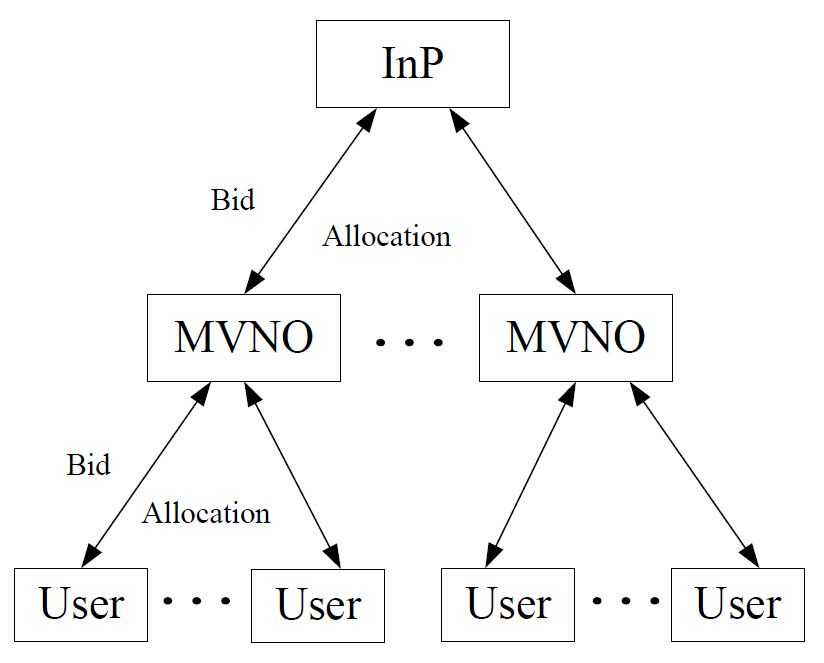}
    \end{center}
    \caption{A single-seller multiple-buyer hierarchical auction model.}
    \label{fig:BasicAuctionModel}
\end{figure}

In the proposed hierarchical auction model, we consider the bidders
to be self-interested each of which carefully chooses the bidding
strategy to maximize her own utility under the given designed
auction mechanism (including allocation and pricing procedures).
Specifically, we consider a quasilinear utility of the following form:
 \beqn
    u_i = v_i - q_i,
\eeqn which represents the difference between valuation and the price charged. The
equilibrium bidding behavior can be analyzed by using noncooperative
game theory (e.g., a Bayesian game considering that the valuations are
private information). Note that there are several desirable
properties when designing an auction mechanism: 1) individual
rationality; 2) incentive compatibility; 3) allocation efficiency.
The definitions for these are given below.

\emph{Definition 1:} An auction mechanism is \emph {individual
rational} if for any bidder $i$, $u_i \geq 0$.

The individual rationality indicates that a bidder will be never
charged more than her valuation of the received resource.

\emph{Definition 2:} An auction mechanism is \emph{incentive
compatible} (truthful) if and only if for every bidder $i$ with true
valuation $v_i$ and for any valuation declaration of the other
bidders $\tilde{b}_{-i}$, we have $u_i(v_i,\tilde{b}_{-i})\geq
u_i(\tilde{b}_i,\tilde{b}_{-i})$, where $\tilde{b}_i \neq v_i$ is
any non-truthful bidding strategy.

That is, truth telling is a dominant strategy for each bidder no
matter how other bidders place their bids. 

\emph{Definition 3:} An allocation is efficient if the sum of
valuations of all accepted bids is maximized.

In the proposed scheme, these properties are achieved by the use of
combinatorial auction~\cite{Vries2003} with appropriate design of
allocation and pricing schemes. Specifically, for the sub-auction
mechanism in each level, we adopt combinatorial auction which allows
the bidders to express preferences over bundles or combinations of
resources (e.g., bundles of subchannels, power, and antennas for a
slice). The consideration of bidding for bundles of resources stems
from the fact that the value of a bundle of items may not be equal
to the sum of individual value of each item in the bundle. That is,
there are possible substitution and complementarity properties among
the items. In this case, the combinatorial auction can lead to more
efficient allocations (in terms of both auction process and outcome)
compared with the case where a traditional single-item auction is
repeated for each item in the bundle.

In a combinatorial auction, for each bidder $i$, the valuation $v_i$
is a mapping from a bundle of items $S_i$ to a real value. A bidder
is a winning bidder and receives value $v_i(S_i)$ if all the items
in the bundle $S_i$ are allocated, otherwise the bidder receives
nothing and the received value is zero (i.e., $v_i(\emptyset)=0$).
The valuation function should also satisfy the monotone property.
That is, for any $S\subseteq T$ we have that $v_i(S)\leq v_i(T)$.
Note that different from single-item auction, there could exist
multiple winning bidders in a combinatorial auction.


The social
welfare obtained by a combinatorial auction is denoted by \beqn V
=\sum v_i(S_i), \eeqn and a socially efficient allocation is an
allocation with maximum social welfare among all allocations such
that \beqn
    V^* = \max \sum v_i(S_i).
\eeqn In our scheme, the use of combinatorial auction can achieve
social efficiency in each level if all bidders bid truthfully, while
pricing schemes are designed for ensuring the incentive
compatibility. The details of incentive compatible pricing schemes
will be introduced in Section III.

\begin{figure}[th]
    \begin{center}
    \includegraphics[width=0.2\textwidth]{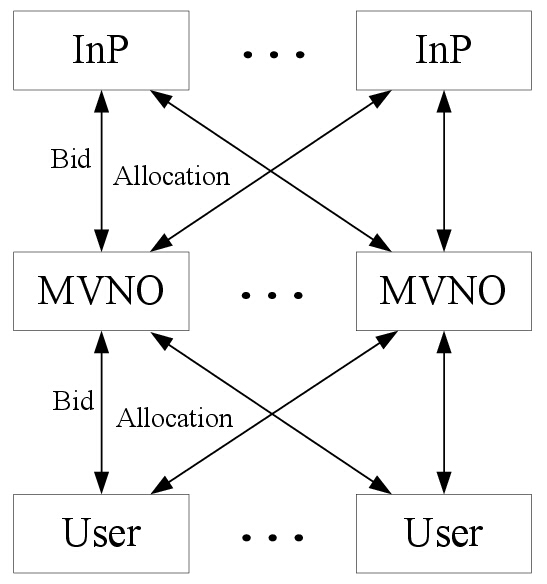}
    \end{center}
    \caption{A multiple-seller multiple-buyer hierarchical auction model.}
    \label{fig:ExtendedAuctionModel}
\end{figure}

2) \emph{Multiple-seller multiple-buyer hierarchical auction model:}
We then extend the above single-seller multiple-buyer model to a
multiple-seller multiple-buyer model as shown in
Fig.~\ref{fig:ExtendedAuctionModel}. The main difference is that
there exists multiple sellers at each level and each bidder could
acquire resources from one of the multiple sellers\footnote{We
assume that a user will not be served simultaneously by multiple
MVNOs.}. Compared with the single-seller model, the multiple-seller
model provides more flexibility for bidders which accounts for
service selection (or user association) and accordingly induces
higher efficiency in resource allocation. For example, in the lower
level, the user association is not fixed for all $K = \sum_k K_m$
users. Instead, a user will be associated with the MVNO that can
satisfy her resource requirement. Accordingly, the competition among
MVNOs is explicitly captured in the sense that the MVNO with larger
amount of resources could attract more users.

It is worth noting that different from the single-seller model, the
MVNOs and InPs are not acting as the auctioneers in the
multiple-seller model. Instead, service brokers are introduced as
external auctioneers at each level on behalf of the sellers (i.e.,
the MVNOs and the InPs). The bidders in each level submits bids to
the corresponding service broker who then determines the resource
allocation and pricing.

\section{Wireless Network Virtualization as a Single-Seller Multiple-buyer Hierarchical Combinatorial Auction}

In this section, we will show how the proposed single-seller
multiple-buyer hierarchical combinatorial auction model can be
applied for virtualization of wireless resources in OFDMA-based 5G
networks with massive MIMO. Specifically, we will answer the
following questions: 1) how do the bidders (i.e., users in the
lower-level auction and MVNOs in the upper-level auction) place
their bids? 2) how to determine the set of winning bids (i.e., the
winner determination problem in the context of combinatorial
auction)? 3) how to solve the WDP in a computationally tractable
manner? 4) how to price the winning bidders such that incentive
compatibility can be achieved? Besides, theoretical analysis of the
properties of the proposed hierarchical auction mechanism is also
provided.

\subsection{How to Place a Bid?}

1) \emph{Bids for users:} We consider two cases.

\textbf{Case I:} We consider that the users explicitly express their
intrinsic physical resource demands (which is the case considered in
most existing work) when applying auction schemes for wireless
resource allocation. In this case, each user explicitly requests a
bundle of resources $S_k$ from her associated MVNO $m$ with \beqn
S_k =\{\mathbf{y}_k, \mathbf{p}_k, A_k\}, \eeqn where $\mathbf{y}_k
=[y_k(n)]^{\hat{C}_m}$, $\mathbf{p}_k = [p_k(n)]^{\hat{C}_m}$, and
$A_k$ are the subchannel request vector, power request vector, and
antenna request, respectively. Note that in this paper, energy
efficiency is not considered. Therefore, all the available antennas
of an MVNO will be activated for achieving the best performance.
Accordingly, we have $A_k = \hat{A}_m, ~\forall k \in
\mathcal{K}_m$.

We assume that the users are single-minded~\cite{Nisan2007} who are
only interested in a specific set of commodities, and the valuation
will be a specified scalar value if the whole set is allocated, and
will be zero otherwise. The definition is given as
follows:

\emph{Definition 4:} A bidder $k$ is single-minded if there is a set
of commodities $S_k$ and a scalar value $a$ such that the valuation
\beqn
    v_k(S) = \left\{\begin{aligned}
    a, ~~~~~~~\forall S\supseteq S_k, \\
    0, ~~~~~ \mbox{otherwise}.
\end{aligned}\right.
\eeqn

For a single-minded user $k$ with explicit resource demand, the bid
can be simply expressed as a pair $\{S_k,b_k(S_k)\}$. The valuation
of the requested resources depends on the application type and could
be different for different users. In this paper, we consider a
linear function $v_k(r_k(S_k)) = \delta_k r_k(S_k)$, where
$\delta_k$ is a constant and $r_k(S_k)$ is the achievable rate if
the requested resource bundle $S_k$ is allocated. The consideration
of valuation functions for different types of applications is
straightforward and will not change the formulation and analysis.
Also, the valuation function considering desirability of the
resources can be designed.

\textbf{Case II:} We consider that each user $k$ implicitly
expresses her resource demand by simply indicating an intrinsic
target rate of $\tilde{R}_k$. How the resources in different
dimensions are allocated to satisfy the rate requirements is left
with the MVNOs who have much stronger computation power.

In this case, the bid pair for user $k$ with implicit resource
demand is expressed as $\{\tilde{R}_k,b_k(\tilde{R}_k))\}$ and the
same valuation function applies for the implicit case.

2) \emph{Bids for MVNOs:} For MVNOs, the bids
are different from that of the users since the MVNOs are
middlemen who do not have intrinsic demands and valuations.
Accordingly, they cannot be considered to be single-minded. Instead,
in the upper-level auction, each MVNO will need to submit bids for
all possible set of resource bundles (i.e., general
valuations)\footnote{In practice, the set of bids submitted can be
restricted considering only available/feasible ones.}. Also, we
consider that the MVNOs explicitly express their physical resource
demands to the InP, and the bid for each MVNO can be expressed as a
combination of pairs $\{S_m,b_m(S_m)\},~\forall S_m \in \Omega_m$,
where $S_m = \{C_m, P_m, A_m\}$ is a tuple representing
the requested number of subchannels $C_m$, the amount of transmit
power $P_m$, and the number of antennas $A_m$, and $\Omega_m =
\{C^{\mathrm{up}},P^{\mathrm{up}},A^{\mathrm{up}}\}$ represents the
set of all possible resource bundles.


Note that although an MNVO can submit combinations of bids, we adopt
XOR-bid with which at most one bid can be accepted for each
MVNO.  For an MVNO $m$ with auctioned resource tuple $S_m$, let us denote
by $q_k(S_m) \geq 0$ the price charged to user $k$ in the
lower-level auction. Then the valuation of resource $S_m$ can be
expressed as \beqn v_m(S_m) = \sum_{k\in \mathcal{K}_m} q_k(S_m) -
q_m^{\mathrm{res}}, \eeqn where $q_m^{\mathrm{res}}$ is the cost for
reserved resources.

A summary of the key elements of the hierarchical auction model is
provided in Table I.

\begin{table}
\caption{A summary of key elements in the proposed hierarchical
auction model} \label{table:summary}
\begin{tabular}{|c|c|c|c|c|c|}
  \hline
   & Demand & Valuation & Bid & Utility & Price \\
  \hline
  $UE_{ex}$ & $S_k$ & $v_k(S_k)$ & $b_k(S_k)$ & $v_k-q_k$ & $q_k$ \\
  \hline
  $UE_{im}$ & $\tilde{R}_k$ & $v_k(\tilde{R}_k)$ & $b_k(\tilde{R}_k)$ & $v_k-q_k$ & $q_k$ \\
  \hline
  MVNO & N.A. & $\sum_{k}q_k-q_m^{\mathrm{res}}$ & $\{b_m(S_m)\}$ & $\sum_{k} q_k-q_m-q_m^{\mathrm{res}}$ & $q_m$ \\
  \hline
  InP & N.A. & N.A. & N.A. & $\sum_{m} q_m$ & N.A. \\
  \hline
\end{tabular}
\end{table}

\subsection{How to Determine the Winning Bids?}

After collecting the bids from the bidders, the auctioneer needs to
determine which set of bids to be accepted. For the proposed
hierarchical auction model, we need to formulate the winner
determination problems (WDPs) for both InP as the auctioneer in the
upper-level auction and MVNOs as auctioneers in the lower-level
auctions.

\subsubsection{WDP formulation for the InP}
Each MVNO $m$ submits a bid combination $\{S_m,b_m(S_m)\}$ for all
possible $S_m \subseteq \Omega_m$ and the objective of the InP is to
maximize the sum value of accepted bids. Accordingly, the WDP for
the InP in the upper-level auction is formulated as
\begin{eqnarray}
    \max_{x_m(S_m)} \sum_{m\in \mathcal{M}} \sum_{S_m \subseteq \Omega_m} b_m(S_m) x_m(S_m) \nonumber \\
    \mbox{s.t.} ~~~~ \sum_{S_m \subseteq \Omega_m} \sum_{m\in \mathcal{M}} x_m(S_m)C_m(S_m) \leq C^{\mathrm{up}},  \\
    \sum_{S_m \subseteq \Omega_m} \sum_{m\in \mathcal{M}} x_m(S_m) A_m(S_m) \leq A^{\mathrm{up}},  \\
    \sum_{S_m \subseteq \Omega_m} \sum_{m\in \mathcal{M}} x_m(S_m) P_m(S_m) \leq P^{\mathrm{up}},  \\
    \sum_{S_m \subseteq \Omega_m} x_m(S_m) \leq 1, \forall~m \in \mathcal{M}, \\
    x_m(S_m) \in \{0,1\}, \forall~S_m, m,
\end{eqnarray}
where $x_m(S_m)$ is a binary variable with $x_m(S_m)=1$ indicating
that the requested resource bundle $S_m$ is accepted and
$x_m(S_m)=0$ otherwise, $C_m(S_m)$, $P_m(S_m)$, and $A_m(S_m)$
represent the number of subchannels, the amount of power, and the
number of antennas requested in $S_m$. The constraints (12)-(14)
ensure that the accepted resource demands do not exceed the
available capacity, while the constraint (15) ensures that no bidder
receives more than one set of resources.

\begin{remark}
Since the power is divisible, having general valuations for each
specific amount of power is impossible. In this case, in the upper
level auction, we consider the power to be discretized based on
certain unit, and each MVNO specifies the number of power units
required in each resource bundle. While such discretization is not
necessary in the lower-level auction since the user is
single-minded.
\end{remark}

\subsubsection{WDP formulation for the MVNOs}

The WDP for each MVNO $m$ in the lower-level auction considering
single-minded users with explicit resource requirements (i.e., Case
I) is formulated as
\beqn
    \max_{x_k} \sum_{k \in \mathcal{K}_m} b_k(S_k) x_k  \nonumber \\
    \mbox{s.t.}~~~
    \sum_{k\in \mathcal{K}_m} p_k(S_k) x_k \leq \hat{P}_m, \\
    \sum_{k\in \mathcal{K}_m}\sum_{S_k\ni n}  x_k \leq J, ~\forall n,  \\
    x_k \in \{0,1\}, ~\forall k,
\eeqn where $x_k$ is the decision variable indicating whether the
resource demand from user $k$ can be satisfied, $p_k(S_k) =
\sum_{n\in \hat{\mathcal{C}}_m}p_k(n)$ is the total power requested
by user $k$, and $S_k\ni n$ if $y_k(n)=1$. The objective is also to
maximize the sum value of accepted bids, and thus achieving
efficient resource allocation as required by wireless
virtualization\footnote{Note that the fairness among users is not
considered.}. The first constraint ensures that the sum power of all
accepted bids does not exceed the available total power. The second
constraint indicates that the number of users sharing a subchannel
cannot exceed $J$.

Note that the WDPs for MVNOs with explicit resource requirement from
users are simpler comparing with that for the InP since the users
are single-minded and therefore do not require general valuations.

Similarly, the WDP for each MVNO $m$ considering single-minded users
with implicit resource requirements (i.e., Case II) is formulated as
\begin{eqnarray}
    \max_{x_k,z_k(\tilde{S}_k)} \sum_{k \in \mathcal{K}_m} b_k(\tilde{R}_k) x_k \nonumber  \\
    \mbox{s.t.}~~~
    r_k(\tilde{S}_k(\tilde{\mathbf{y}}_k,\tilde{\mathbf{p}}_k,A_m)) = \tilde{R}_k, ~\forall \tilde{S}_k \in \Omega_k, \label{eq:rate_con} \\
    \sum_{k\in \mathcal{K}_m} \sum_{\tilde{S}_k \in \Omega_k} \sum_{\tilde{S}_k\ni n}  z_k(\tilde{S}_k) \leq J, ~\forall n,  \label{eq:user_con} \\
    \sum_{k\in \mathcal{K}_m} \sum_{\tilde{S}_k\in \Omega_k} p_k(\tilde{S}_k) z_k(\tilde{S}_k) \leq \hat{P}_m, \\
    \sum_{\tilde{S}_k \in \Omega_k} z_k(\tilde{S}_k) \leq 1,~\forall k,  \label{eq:exclusive_con} \\
    x_k, z_k(\tilde{S}_k) \in \{0,1\},
\end{eqnarray}
where $x_k=1$ indicates the rate demand of user $k$ can be
satisfied, and $x_k=0$ otherwise, $\tilde{S}_k =
\{\tilde{\mathbf{y}}_k,\tilde{\mathbf{p}}_k,A_m\}$ is the tuple
representing a resource allocation profile considering the freedoms
in power, frequency, and spatial domains. Specifically,
$\tilde{\mathbf{y}}_k = [\tilde{y}_k(n)]^{\hat{C}_m}$ is the
subchannel allocation vector with $\tilde{y}_k(n)=1$ indicating
subchannel $k$ is assigned to user $k$. $\tilde{\mathbf{p}}_k =
[p_k(n)]^{\hat{C}_m}$ is the power allocation vector.
Equation~(\ref{eq:rate_con}) is the rate constraint which requires
the resource allocation profile $\tilde{S}_k$ should satisfy the
user rate requirement\footnote{Note that if minimum rate requirement
is considered, the `$=$' in (\ref{eq:rate_con}) will be replaced by
`$\geq$'. The use of minimum rate requirement will not change the
analysis since it only increases the number of allocation strategies
for satisfying the rate requirement.}. The achievable rate can be
calculated from (\ref{eq:DLRate}). The set of all resource
allocation profiles $\tilde{S}_k$ for user $k$ constitutes the set
$\Omega_k$. $z_k(S_k)$ is an indicator with value $1$ indicating the
allocation profile $S_k$ is accepted, and $0$ otherwise. Apparently
$x_k =1$ if there exists $z_k(\tilde{S}_k) = 1$ for all
$\tilde{S}_k$. Equation (\ref{eq:user_con}) indicates that the
number of users sharing a subchannel cannot exceed $J$. Equation
(\ref{eq:exclusive_con}) indicates that for a user $k$, at most one
allocation profile $\tilde{S}_k \in \Omega_k$ can be accepted.

\begin{remark}
It can be observed that the above optimization problem is
fundamental for resource allocation which jointly considers
\emph{feasibility}, \emph{admission control}, and \emph{allocation}
problems in a unified framework. Specifically, given certain degrees
of freedom in resource allocation, the solution of the above problem
gives the answers on whether the system can satisfy the rate
requirements of all users (i.e., feasibility problem). And if not,
which users should be accepted (i.e., admission control problem),
and how the resources can be allocated so that the sum utility of
all users is maximized (i.e., allocation problem). 
\end{remark}

\subsection{How to Solve the WDPs in the Hierarchical Auction?}

Similar to solving a hierarchical game (e.g., Stackelberg game), the
method of backward induction can be used for solving the proposed
hierarchical auction problem. In this case, we start with solving
the WDP for an MVNO $m$ in the lower-level auction.

\subsubsection{Solving the WDP in the lower-level auction}

It has been shown in the literature that the winner determination
problem of combinatorial auctions is an integer programming problem
which is NP-hard even for the single-minded
cases~\cite{Sandholm2002}. For simplifying the original problem, we
consider two assumptions as follows:

\emph{Assumption 1:} We assume the subchannels to be homogeneous for
each end user (i.e., the channel gains of different subchannels are
the same for a user, while they can be different for different
users). Accordingly, equal power is allocated to each assigned
subchannel.

\emph{Assumption 2:} The achievable rate of a user is independent of
which other users are sharing the same channel.

These assumptions are practical for massive MIMO systems since the
small-scale fading are averaged out and only the large-scale fading
(e..g, path-loss and shadowing) affects. Accordingly, the SINR is
constant with respect to frequency since the slow fading
coefficients are independent of frequency~\cite{Marzetta2010}. Also,
the users' channels decorrelate with the increasing number of BS
antennas~\cite{Bjornson2014}.

Based on these two assumptions, the WDP for MVNO with explicit user
resource request can be reformulated as
\beqn
    \max_{x_k} \sum_{k \in \mathcal{K}_m} b_k(S_k) x_k \nonumber \\
    \mbox{s.t.}~~~
    \sum_{k\in \mathcal{K}_m} p_k(S_k) x_k \leq \hat{P}_m,    \\
    \sum_{k \in \mathcal{K}_m} c_k(S_k) x_k \leq \hat{C}_m J,   \\
    x_k \in \{0,1\}, ~\forall k,
\eeqn where $c_k \in [0,\tilde{C}_m]$ and $p_k \in [0,\tilde{P}_m]$
represent the requested number of subchannels and power,
respectively. Note that in this reformulation, we can use $c_k$ and
$p_k$ instead of $\mathbf{y}_k$ and $\mathbf{p}_k$ since we consider
the channels to be homogeneous. Therefore, the expression for the resource
bundle requested by user $k$ is simplified as $S_k=\{c_k,p_k,A_m\}$.

It can be seen that the simplified problem is also an integer
program which is still NP hard. Accordingly, there are fundamental
tradeoffs between the social efficiency (optimality) and
computational complexity. In general, there are two possible ways
for solving the reformulated WDP in a computationally tractable
manner. The first is to find the exact optimal solution for a small
problem (e.g., through dynamic programming or branch-and-bound
method) at the cost of possible high computational complexity. The
second is to design low-complexity algorithms to find approximate
optimal solutions for large scale problems.

We first propose a dynamic programming-based algorithm for obtaining
the exact solution. The main idea is to divide the original problem
$\mathrm{WDP}_m(K_m,\hat{P}_m,\hat{C}_m)$ into similar sub-problems
which can be solved recursively. Specifically, we partition the
resource allocation into $K_m$ stages, and denote by
$\mathrm{WDP}_m(k,\mathbf{e}(k))$ the subproblem which considers the
resource allocation to $k$ users with available resources
$\mathbf{e}(k) = [e_c(k),e_p(k)]^\mathrm{T}$, where $e_c(k)$ and
$e_p(k)$ denote, respectively, the available subchannels and power at stage $k$,
which can also be regarded as the state variables. In
each stage $k$, the MVNO $m$ allocates $u_c(k)$ subchannels and
$u_p(k)$ units of power to user $k$. Denote by $\mathbf{u}(k) =
[u_c(k),u_p(k)]^\mathrm{T}$. Accordingly, the state transition equation
can be expressed as \beqn
    \mathbf{e}(k+1) = \mathbf{e}(k) - \mathbf{u}(k). \nonumber
\eeqn Also, denote by $\mathbf{x}^*(k,\mathbf{e}(k)) =
[x_1^*,\ldots,x_k^*]$ the optimal solution to the subproblem
$\mathrm{WDP}_m(k,\mathbf{e}(k))$ with the corresponding optimal value
$f(k,\mathbf{e}(k))$. Accordingly, we can have 
$f(k,\mathbf{e}(k)) =$

$\max \{f(k-1,\mathbf{e}(k)), f(k-1,\mathbf{e}(k)-\mathbf{u}(k)) +
b_k(S_k)\}$, 
for $k=2,\ldots,K_m$. The initial condition is given by
\beqn
    f(1,\mathbf{e}(1)) = \left\{\begin{aligned}
    b_1(S_1), ~e_c(1)\geq c_1(S_1), e_p(1)\geq p_1(S_1),  \nonumber \\
    -\infty, ~~~~~~~~~~~~~~~~~~~~~~~~~~  \mbox{otherwise}. \nonumber
\end{aligned}\right.
\eeqn

Note that for single-minded users, we only need to consider the
state transition for fixed $u_c(k)$ and $u_p(k)$ at each state.
Specifically, $u_c(k) = c_k(S_k)$ and $u_p(k) = c_p(S_k)$ if $x_k=1$,
and $u_c(k) = 0$ and $u_p(k) = 0$ if $x_k=0$.
The details  are given in \textbf{Algorithm 1}.

\begin{algorithm}
\caption{A dynamic programming-based algorithm to solve the WDP
for MVNO with explicit user resource request}
\begin{algorithmic}
   \STATE 1. Initialization: collect $b_k(S_k)$ from each user $k$, and calculate initialization condition $f(1,\mathbf{e}(1))$.
   \STATE 2. For each stage and each possible state calculate the optimal value function $f(k,\mathbf{e}(k))$.
   \STATE 3. Output: Find $f(K_m,\hat{C}_m,\hat{P}_m)$ and obtain the corresponding optimal allocation in
   each stage by using \beqn
    x^*_k = \arg \max_{x_k} \{f(k-1,\mathbf{e}(k)),
    f(k-1,\mathbf{e}(k)-\mathbf{u}(k)) + b_k(S_k)\}. \nonumber
\eeqn
\end{algorithmic}
\end{algorithm}

Compared with exhaustive enumeration with time-complexity of $O(2^{K_m})$, the time-complexity of the dynamic programming-based
algorithm is of $O(K_m \Theta_m^2)$, where $\Theta_m =
\max\{\hat{C}_m,\hat{P}_m\}$. The significant reduction in
time-complexity stems from the fact that the optimal value for each stage
and each state is stored and calculated only once, while it needs to
be calculated repeatedly in exhaustive enumeration.

We also implement a polynomial-time greedy algorithm to obtain
an approximate optimal solution which is based on the algorithm
proposed in~\cite{Lehmann2002},~\cite{Zaman2013} for WDP with
single-minded bidders. The main idea of the greedy algorithm is to
allocate the resource to bidders with larger normalized value.
Specifically, after collecting all the bids from the users, the MVNO
as the auctioneer sorts the bids in a decreasing order of
$\frac{b_k}{\sqrt{|S_k|}}$ which is viewed as the normalized value
of a bid. Note that since the requested resource bundle $S_k$
consists of multiple dimensional resources, we consider $|S_k| =
\omega_c c_k + \omega_p p_k$, which is a weighted sum of the number
of different type of resources requested. The details of the greedy
algorithm are presented in \textbf{Algorithm 2}.

\begin{algorithm}
\caption{A greedy algorithm to solve the WDP for an MVNO with
explicit user resource request}
\begin{algorithmic}
   \STATE 1. Initialization: set $c=0$ and $p=0$.
   \STATE 2. For each submitted bid pair $\{S_k, b_k(S_k)\}$, calculate
   $b_k(S_k)/\sqrt{|S_k|}$. Re-index all bid pairs such that
   \beqn
        \frac{b_1(S_1)}{\sqrt{|S_1|}} \geq
        \frac{b_2(S_2)}{\sqrt{|S_2|}} \geq \cdots \geq
        \frac{b_{K_m}(S_{K_m})}{\sqrt{|S_{K_m}|}}.
   \eeqn
   \STATE 3. For $k=1:K_m$, if $c + c_k \leq \hat{C}_m J$ and $p + p_k \leq \hat{P}_m$,
   then allocate $c_k$ number of subchannels and $p_k$ units of power to corresponding
   user $k$.
\end{algorithmic}
\end{algorithm}

Note that with this algorithm we may only obtain an approximate
optimal solution while the dynamic programming algorithm can obtain
the exact solution. This difference will have impacts on the choice
of pricing scheme to guarantee the incentive compatibility which
will be shown next.

For the WDP of MVNO with implicit user resource request, although the
users are single-minded, there exist combinations of resource
allocation strategies for achieving the target rate due to the
freedoms in multiple domains. Accordingly, it is equivalent to a WDP
formulation with general valuations. Given assumptions 1 and 2, the
WDP for the MVNO with implicit user resource request can be
reformulated as a WDP with general valuations as follows:
\begin{eqnarray}
    \max \sum_{k \in \mathcal{K}_m} \sum_{\tilde{S}_k \in \Omega_k} b_k(\tilde{R}_k) x_k(\tilde{S}_k) \nonumber \\
    \mbox{s.t.}~~~
    r_k(\tilde{S}_k(\tilde{c}_k,\tilde{p}_k,A_m)) = \tilde{R}_k, ~\forall \tilde{S}_k \in \Omega_k, \\
    \sum_{k\in \mathcal{K}_m} \sum_{\tilde{S}_k\in \Omega_k} p_k(\tilde{S}_k) x_k(\tilde{S}_k) \leq P_m, \\
    \sum_{k \in \mathcal{K}_m} \sum_{\tilde{S}_k\in \Omega_k} c_k(\tilde{S}_k) x_k(\tilde{S}_k) \leq \hat{C}_m J, \\
    \sum_{\tilde{S}_k\in \Omega_k} x_k(\tilde{S}_k) \leq 1,~\forall k, \\
    x_k(\tilde{S}_k) \in \{0,1\}.
\end{eqnarray}

To solve the above problem, we first need to find the set of
resource allocation strategies (i.e., $\Omega_k$) which satisfy the
rate requirement. Specifically, all the antennas available will be
activated. Also, we assume the channels to be homogeneous for a
user. In this case, for each number of requested subchannels, the
power required for satisfying the rate requirement can be calculated
according to~(\ref{eq:DLRate}). After obtaining the set $\Omega_k$,
we extend the previous dynamic programming algorithm and greedy
algorithm for single-minded bidders to accommodate the general
valuation case.

Specifically, for the dynamic programming algorithm, the main
difference is that when calculating the optimal value function, it
is required to consider all possible state transitions due to the
general valuations. Specifically, \beqn f(k,\mathbf{e}(k)) =
\max_{\mathbf{u}(k)} \{b_k(\mathbf{u}(k)) +
f(k-1,\mathbf{e}(k)-\mathbf{u}(k))\}, \nonumber \eeqn where $u_c(k)
\in [0,\tilde{C}_m]$ and $u_p(k) = p_k(u_c(k))$.

To extend the greedy algorithm  for the WDP with general XOR bid, we
consider the bid combinations
$\{\tilde{S}_k,b_k(\tilde{S}_k)\}\footnote{Note that
$b_k(\tilde{S}_k)$ is equivalent to $b_k(\tilde{R}_k)$ which can be
simply obtained as a linear function of $\tilde{R}_k$.}, \forall
\tilde{S}_k \in \Omega_k$ submitted by user $k$ as combinations of
bid pair $\{\tilde{S}_k, b_k(\tilde{S}_k)\}$ submitted by virtual
single-minded bidders the number of which is equal to the number of
all possible bid combinations. For example, there exists
$|\Omega_k|$ number of virtual bidders with user $k$. Note that due
to the XOR bid, at most one virtual single-minded bidder of each
user can be accepted. To address this problem, we introduce the
concept of virtual commodity corresponding to a user. The bid
combinations of user $k$ are extended as $\{\tilde{S}_k \bigcup k,
b_k(\tilde{S}_k)\}$. Since the virtual commodity $k$ can only be
allocated to one winning bidder, the original XOR bid can be
realized. The extended greedy algorithm is given in
\textbf{Algorithm 3}.

\begin{algorithm}
\caption{A greedy algorithm to solve the WDP for an MVNO with
implicit user resource request}
\begin{algorithmic}
   \STATE 1. Initialization: set $c=0$, $p=0$ and $x_k=0$ for each user $k$.
   \STATE 2. For each submitted bid pair $(\tilde{S}_k \bigcup k, b_k(\tilde{S}_k))$, calculate
   $b_k(\tilde{S}_k)/\sqrt{|\tilde{S}_k|}$. Re-index all bid pairs such that
   \beqn
        \frac{b_1(\tilde{S}_1)}{\sqrt{|\tilde{S}_1|}} \geq
        \frac{b_2(\tilde{S}_2)}{\sqrt{|\tilde{S}_2|}} \geq \cdots \geq
        \frac{b_T(\tilde{S}_T)}{\sqrt{|\tilde{S}_T|}},
   \eeqn where $T = \sum_k |\Omega_k|$.
   \STATE 3. For $k=1:T$, if $c + c_k(\tilde{S}_k) \leq \hat{C}_m J$, $p + p_k(S_k) \leq
   \hat{P}_m$, and $x_k = 0$, then allocate $c_k(\tilde{S}_k)$ number of subchannels and $p_k(S_k)$ amount of power to corresponding
   user $k$ and set $x_k = 1$.
\end{algorithmic}
\end{algorithm}

\subsubsection{Solving the WDP in the upper-level auction}


Compared with users in the lower-level auction, the MVNOs are not
single-minded which results in combinations of XOR bids. Note that
the valuation of each resource bundle in the upper-level auction
depends on the resale gain in the lower-level auction, which can be
obtained by solving the lower-level auction supposing this resource
bundle is allocated. After obtaining the valuations of all resource
bundles, similar algorithms for solving the WDP with implicit
resource request can be applied.  For example, with the
dynamic programming-based algorithm, the resource allocation is
partitioned into $M$ stages, and the optimal value function of each
stage is \beqn f(m,\mathbf{e}(m)) = \max_{\mathbf{u}(m)}
\{b_m(\mathbf{u}(m)) + f(m-1,\mathbf{e}(m)-\mathbf{u}(m))\},
\nonumber
 \eeqn with $\mathbf{e}(m) = [e_c(m),e_p(m),e_a(m)]^\mathrm{T}$ and
$\mathbf{u}(m) = [u_c(m),u_p(m),u_a(m)]^\mathrm{T}$, where $u_c(m)
\in [0,\tilde{C}^{\mathrm{up}}]$, $u_p(m) \in
[0,\tilde{P}^{\mathrm{up}}]$, and $u_a(m) \in
[0,\tilde{A}^{\mathrm{up}}]$.

Note that in a practical situation the number of MVNOs may not be large
(e.g,. $m=3,4$). Also, it is reasonable that the InP sells the
resources to MVNOs only in a grouped manner (e.g., a group of 5 or
10 subchannels). Furthermore, the InP could impose restrictions on the
maximum allowable groups of subchannels each MVNO can bid for.
Accordingly, the bid combinations can be significantly reduced and
the complexity of finding the exact optimal solutions can be
reduced. Note that, apparently, restricting the bid combinations
would also incur a tradeoff between the computational complexity and
social efficiency.

\subsection{How to Price the Winning Bidders?}

\subsubsection{Pricing scheme with exact solution for the WDP}

The design of pricing scheme is of critical importance for achieving
incentive compatibility. With all bidders bidding truthfully, the
above WDPs which aim to maximize the sum of accepted bids can also
achieve social optimality at each level since $b_k = v_k$.

For single-commodity auctions, the second-price auction (i.e.,
Vickrey auction~\cite{Vickrey1961}) has been shown to be an
incentive compatible scheme with which the winning bidder (with
highest bid) pays the second highest bid. The VCG
scheme~\cite{Vickrey1961,Clarke1971,Groves1973}, as a generalization
of the second-price auction to multiple commodities, preserves the
incentive compatibility. The intuitive idea of VCG pricing is that a
bidder should pay the potential loss they impose to other bidders.
Specifically, with VCG pricing, a bidder $k$ will be charged
\beqn
    q_k^{\mathrm{vcg}} = \sum_{j \neq k} v_j(\bar{S}^*_j) - \sum_{j \neq k}
    v_j(S^*_j),
\label{eq:VCGPricing} \eeqn where $\bar{S}^*_j$ and $S^*_j$
represent, respectively, the resources obtained by bidder $j$ when bidder $k$ is
not participating and is participating. For a winning
bidder, the VCG price is the decrease of welfare of all other
bidders caused by her presence. Note that the VCG prices are
nonnegative.

Although the VCG pricing can achieve incentive compatibility, it is
not designed for maximizing the seller's revenue. In some cases, the
resulting revenue (e.g., $\sum q_k^{\mathrm{vcg}}$) can be even far
from the optimal one. For example, if there are sufficient resources
such that the requirements of all bidders can be satisfied, the VCG
price for each bidder is zero. This could cause problem for MVNOs whose valuations depend on
the revenue gained from resale. Specifically, with VCG pricing, the
valuation $v_m(S_m)$ could decrease with an increasing number of
resources in $S_m$ which may motivate the MVNO to lease less
resources.

To address this problem while guaranteeing the incentive
compatibility, we jointly use the VCG pricing with a base access
price. Specifically, each type of resource has a base access price,
and a user who is admitted (i.e., a successful bidder) will be
charged the larger of the base price and the VCG price. That is
\beqn
    q_k = \max\{q_k^{\mathrm{base}}, q_k^{\mathrm{vcg}}\}.
\eeqn

Note that the base price is known to all users. If a bidder is aware
that the valuation of the requested resources is even less than the
base price, she will not place the bid. In this pricing scheme, the
base price can guarantee certain revenue for the MVNO, while the VCG
price could represent the impact of satisfying the resource
requirement of a user to the other users. For example, if there are
sufficient resources, the VCG price for a bidder can be low.
However, when there is an intense resource competition, the more
resource a user requests, the higher will be the probability that
other users will not be admitted, and the higher will be the VCG
price charged to this user.

Although the objective of each auctioneer is to maximize the social
welfare, with such a pricing scheme, we can also achieve approximate
optimal seller's revenue. Also, the incentive compatibility can be
preserved.

\subsubsection{Pricing scheme with approximate solution for the WDP}

The VCG pricing can preserve incentive compatibility only if the WDP
is solved exactly (i.e., optimal solution is obtained), while it is
incompatible with approximate algorithms in general\cite{Nisan2007}.
That is, if we obtain sub-optimal $\bar{S}^*$ and $S^*$
in~(\ref{eq:VCGPricing}) using approximate algorithms, the
corresponding VCG price $q_k$ is not incentive compatible.

Accordingly, the corresponding pricing scheme needs to be designed for
the greedy algorithm. Specifically, we propose to jointly use the
base access price with a VCG-like pricing. We first give the
definition of {\em blocking} as follows:

\emph{Definition 5:} Assume bidder $k$ with bid $b_k$ is a winning
bidder while a bidder $j$ with bid $b_j$ is not accepted. Then the
bidder $k$ uniquely blocks bidder $j$ if the bidder $j$ is a winning
bidder without bidder $k$'s participation in the auction.

Denote by $\mathbb{B}_k$ the set of bidders blocked by bidder $k$.
The main idea is to charge bidder $k$ according to the highest value
bid it blocks. Note that this value is also normalized as that in
the greedy algorithm. The VCG-like price for a winning bidder
$k$ is \beqn
    q_k^{\mathrm{greedy}} =
    \max_{j\in{\mathbb{B}_k}}\frac{b_j}{\sqrt{|S_j|}}\sqrt{|S_k|}.
\eeqn Accordingly, the winning bidder will be charged the larger of
the base access price and the VCG-like price as follows: \beqn
    q_k = \max\{q_k^{\mathrm{base}}, q_k^{\mathrm{greedy}}\}.
\eeqn

These two pricing schemes can be applied in both upper level and
lower level auctions in accordance with the solving algorithms used.

\subsection{Analysis of Properties of the Proposed Hierarchical Auction Mechanism}

In this part, we analyze the properties of the proposed hierarchical
auction mechanism. We first show that the proposed mechanism can
achieve individual rationality.

\begin{theorem}
The proposed hierarchical auction mechanism is individual rational
for all truthful bidders in both upper and lower level auctions.
\end{theorem}
\begin{proof}
The VCG scheme together with exact WDP solving algorithms has been
shown to be individual rational~\cite{Groves1973}. For the greedy
algorithm with corresponding pricing scheme, we consider two cases.
First, if there is no bidder blocked by a winning bidder $k$ (i.e.,
$\mathbb{B}_k = \emptyset$), then $q_k = \max\{q_k^{base},0\} =
q_k^{base}$. Then, $b_k \geq q_k^{base}$ and accordingly
$u_i \geq 0$. Second, if the set $\mathbb{B}_k \neq \emptyset$,
according to the pricing scheme, the price charged for bidder $k$ is
\beqn
    q_k = \max\{q_k^{base}, \max_{j\in \mathbb{B}_k}\frac{b_j}{\sqrt{|S_j|}}\sqrt{|S_k|}\}.
\eeqn While according to the allocation scheme, we have \beqn
\frac{b_k}{\sqrt{|S_k|}} \geq \max_{j\in \mathbb{B}_k}
\frac{b_j}{\sqrt{|S_j|}}.  \eeqn Accordingly, we can have $b_k \geq
q_k$ and $u_k \geq 0$.
\end{proof}

We also show the property of allocation efficiency in the following
theorem.

\begin{theorem}
With the proposed
dynamic programming algorithms for exact solution of the WDPs, the proposed hierarchical auction mechanism achieves allocation
efficiency with truthtelling bidders at each level.
\end{theorem}

This result can be obtained immediately from the WDP formulation
which aims to maximize the sum of accepted bids. Note that the property of
allocation efficiency  is not preserved for the entire
hierarchial auction which will be shown in the numerical results.
Also, similar observation was made in~\cite{Tang2012}.

In the following, we will analyze the incentive compatibility of the
proposed mechanism. To this end, we first introduce the concepts of
\emph{monotone} and \emph{critical value} as follows:

\emph{Definition 6:} The allocation scheme of an auction is
\emph{monotone} if a bidder $k$ with bid $\{b_k(S_k),S_k\}$ is a
winning bidder, then all bidders $j$ with $\{b_j(S_j),S_j\} \succeq
\{b_k(S_k),S_k\}$ are also winning bidders.

The notation $\succeq$ denotes the preference over bid pairs.
Specifically, $\{b_j(S_j),S_j\} \succeq \{b_k(S_k),S_k\}$ if
$b_j(S_j)\geq b_k(S_k)$ for $|S_j| = |S_k|$ or $|S_j| \leq |S_k|$
for $b_j(S_j) = b_k(S_k)$. The monotonicity indicates that the
chance for obtaining a required bundle of resources can only be
increased by either increasing the bid or decreasing the amount of
resources required.

\emph{Definition 7:} For a given monotone allocation scheme, there
exists a \emph{critical value} $\hat{q}_k$ of each bid pair
$\{b_k(S_k),S_k\}$ such that $\forall b_k \geq \hat{q}_k$ will be a
winning bid, while $\forall b_k < \hat{q}_k$ is a losing bid.

The critical value can be seen as the minimum a bidder has to bid
for obtaining the requested bundle of resources. With the concepts
of \emph{monotonicity} and \emph{critical value}, we have the
following lemma~\cite{Mualem2002}:

\begin{lemma}
An auction mechanism is incentive compatible if the allocation
scheme is monotone and each winning bidder pays the critical value.
\end{lemma}

The VCG scheme has been shown to be incentive
compatible~\cite{Groves1973} for allocation algorithms which solve
the WDP exactly. Accordingly, we can have the following theorem.

\begin{theorem}
For the sub-auction at each level of the proposed hierarchical
auction, the mechanism consisting of proposed dynamic programming
algorithm for WDP and VCG pricing with base access price achieves
incentive compatibility (for both single-minded and general
valuation cases).
\end{theorem}

We will show that for single-minded users with explicit resource
request, the greedy algorithm with corresponding designed pricing
scheme can also achieve incentive compatibility in the lower-level
auction.

\begin{theorem}
The proposed auction mechanism with greedy algorithm is monotone for
the case of single-minded users with explicit resource request. The
corresponding pricing scheme charges the winning bidders their
critical values and the corresponding auction mechanism is
incentive compatible.
\end{theorem}
\begin{proof}
The proof of monotonicity can be immediately obtained from the
allocation algorithm. Specifically, a bidder can increase her order
in the ranking by either increasing the bid value or reducing the
amount of required resources. For example, two users requesting the
same number of subchannels and power, the user with higher
achievable rate will be preferred.

Then we will find the critical value which is the minimum a bidder
has to bid to win the requested bundle of resource. Denote by $j$
the blocked bidder with highest normalized valuation who would win
if bidder $k$ is not participating in the auction. Accordingly, the
minimum bid the bidder $k$ needs to place is
$\frac{b_j}{\sqrt{|S_j|}}\sqrt{|S_k|}$, which is just the payment of
bidder $k$ in the pricing scheme.

With monotonicity and critical payment property, it is
straightforward that the incentive compatibility can be achieved
according to Lemma 3.3.
\end{proof}

However, note that the result on incentive compatibility for the greedy
algorithm and pricing is only valid for single-minded cases and
it will not hold for bidders with general valuations. Although there
exist schemes which preserve incentive compatibility for general
valuations, the worst-case performance is much inferior than that for the greedy
algorithm.

Regarding the incentive compatibility of the entire hierarchical
auction mechanism, we have the following theorem.
\begin{theorem}
The proposed hierarchical auction mechanism is incentive
compatible with any combination of incentive compatible sub-auctions
at each level.
\end{theorem}

The proof is similar to that in~\cite{Tang2012} and is omitted here.

\section{Extension to a Multiple-Seller Multiple-Buyer Hierarchical Auction Model}

We now extend the single-seller multiple-buyer hierarchical auction
to a multiple-seller multiple-buyer model. In this case, the users
are not restricted to only one provider but can freely choose among
several MVNOs. Similarly, the MVNOs can choose from different InPs.

\subsection{WDP Formulations}

The bids are the same as that in the single-seller
model. The WDP formulation for the service broker as the auctioneer
in the upper-level auction is expressed as
\begin{eqnarray}
    \max_{x_{im}(S_m)} \sum_{i \in \mathcal{I}} \sum_{m\in \mathcal{M}} \sum_{S_m} b_m(S_m) x_{im}(S_m) \nonumber \\
    \mbox{s.t.}   \sum_{S_m} \sum_{m\in \mathcal{M}} x_{im}(S_m) C_m(S_m) \leq C_i -\sum_{m \in \mathcal{M}} C_{im}^{res}, ~\forall~i, \nonumber  \\
    \sum_{S_m} \sum_{m\in \mathcal{M}} x_{im}(S_m) A_m(S_m) \leq A_i -\sum_{m \in \mathcal{M}} A_{im}^{res}, ~\forall~i, \nonumber \\
    \sum_{S_m} \sum_{m\in \mathcal{M}} x_{im}(S_m) P_m(S_m) \leq P_i -\sum_{m \in \mathcal{M}} P_{im}^{res}, ~\forall~i, \nonumber \\
        \sum_{i \in \mathcal{I}} \sum_{S_m} x_m(S_m) \leq 1, \forall~m \in \mathcal{M}, \nonumber \\
         x_{im}(S_m) \in \{0,1\}, \forall~S_m , m, i, \nonumber
\end{eqnarray}
where $\mathcal{I}$ is the set of InPs and $x_{im}(S_m)=1$ indicates
the resource request for $S_m$ from MVNO $m$ to InP $i$ is accepted.
Note that the MVNO can only lease resources from one of the InPs.

Similarly, the WDP for the service broker in the lower-level auction
with explicit user resource request is expressed as: \beqn
    \max_{x_{mk}} \sum_{m\in \mathcal{M}} \sum_{k\in \mathcal{K}} b_k(S_k)
    x_{mk} \nonumber \\
    \mbox{s.t.}~~~
    \sum_{k\in \mathcal{K}} p_k(S_k) x_{mk} \leq \hat{P}_m,~\forall m, \nonumber \\
    \sum_{k\in \mathcal{K}} c_k(S_k) x_{mk} \leq \hat{C}_m J_m,~\forall m, \nonumber \\
    \sum_{m \in \mathcal{M}} x_{mk} \leq 1,~\forall k, \nonumber \\
    x_{mk} \in \{0,1\}, ~\forall m,k, \nonumber
\eeqn where $\mathcal{K} = \bigcup \mathcal{K}_m$.

The WDP for the service broker in the lower-level auction with
implicit user resource request is expressed as:
\begin{eqnarray}
    \max_{x_{mk}(\tilde{S}_{mk})} \sum_{m \in \mathcal{M}} \sum_{k \in \mathcal{K}} \sum_{\tilde{S}_{mk} \in \Omega_{mk}} b_k(\tilde{S}_{mk}) x_{mk}(\tilde{S}_{mk}) \nonumber \\
    \mbox{s.t.}~~~
    r_k(\tilde{S}_{mk}) = \tilde{R}_k, ~\forall \tilde{S}_{mk} \in \Omega_{mk},    \nonumber \\
    \sum_{k\in \mathcal{K}} \sum_{\tilde{S}_{mk}} p_k(\tilde{S}_{mk}) x_{mk}(\tilde{S}_{mk}) \leq P_m, ~\forall ~m, \nonumber \\
    c_k(\tilde{S}_{mk}) \leq \tilde{C}_m, ~\forall m, \nonumber \\
    \sum_{k\in \mathcal{K}} c_k(\tilde{S}_{mk}) x_{mk}(\tilde{S}_{mk}) \leq
    \hat{C}_m J_m, ~\forall m, \nonumber \\
    \sum_{m\in \mathcal{M}} \sum_{\tilde{S}_{mk}} x_{mk}(\tilde{S}_{mk}) \leq 1,~\forall k, \nonumber \\
    x_{mk}(\tilde{S}) \in \{0,1\}. \nonumber
\end{eqnarray}

\subsection{Allocation and Pricing Schemes}

The above WDPs are equivalent to multiple multidimensional knapsack
problems. For such problems, the dynamic programming-based methods
require huge memory~\cite{Kellerer2004}. In this case,
branch-and-bound method can be applied to find the exact solutions.
The key challenge for applying branch-and-bound approaches is to
find a tight upper bound of the problem with which then standard
branch-and-bound algorithms can be used (e.g.,
~\cite{Martello1990}). Therefore, in this part, we will focus on the
derivation of the upper bound of the WDPs.

We use the WDP in the lower-level auction with explicit user
resource request as an example. To obtain the upper bound, surrogate
relaxation is used\footnote{Note that continuous relaxation is not
used here since the upper bound obtained by continuous relaxation is
dominated by the bound obtained by surrogate relaxation.}. Given a
set of positive vector of multipliers $\pi$, the standard surrogate
relaxation of the original WDP problem can be expressed as \beqn
    \max_{x_{mk}} ~~S = \sum_{m\in \mathcal{M}} \sum_{k\in \mathcal{K}} b_k(S_k)
    x_{mk} \nonumber \\
    \mbox{s.t.}~~~
    \sum_{m\in \mathcal{M}} \pi_m \sum_{k\in \mathcal{K}} p_k(S_k)
    x_{mk} \leq \sum_{m\in \mathcal{M}} P_m, \nonumber \\
    \sum_{m\in \mathcal{M}} \pi_m \sum_{k\in \mathcal{K}} c_k(S_k)
    x_{mk} \leq \sum_{m\in \mathcal{M}} \hat{C}_m J_m,  \nonumber \\
    \sum_{m \in \mathcal{M}} x_{mk} \leq 1,~\forall k, \nonumber \\
    x_{mk} \in \{0,1\}, ~\forall m,k. \nonumber
\eeqn

Denote by $\tilde{o}$ the optimal value of the above relaxed
problem. Accordingly, $\tilde{o}$ is an upper bound of the original
problem for arbitrary nonnegative multiplier vector of $\pi$. To
achieve a tight upper bound, the optimal multiplier should be
chosen such that $S$ is minimized. That is  \beqn \pi^* =
\arg \min_{\pi} S(\pi). \eeqn

Compared with Lagrangian relaxation the optimal multipliers of
which can only be obtained numerically (e.g., through subgradient
methods), the optimal value of the multipliers for surrogate
relaxation for the formulated problem can be easily obtained as
shown in the following lemma~\cite{Martello1990}.

\begin{lemma}
For any instance of multiple knapsack problem, the optimal vector of
multipliers for $S(\pi)$ is $\pi_m^* = \zeta$ for all $m$, where
$\zeta$ is any positive constant.
\end{lemma}

With the optimal vector of multipliers, the relaxed problem becomes
\beqn
    \max \sum_{k\in \mathcal{K}} b_k \acute{x}_k  \nonumber \\
    \mbox{s.t.}~~~
    \sum_{k \in \mathcal{K}} p_k(S_k) \acute{x}_k \leq \sum_{m \in
    \mathcal{M}} P_m, \nonumber \\
    \sum_{k \in \mathcal{K}} c_k(S_k) \acute{x}_k \leq \sum_{m \in
    \mathcal{M}} \hat{C}_m J_m, \nonumber \\
    \acute{x}_k \in \{0,1\},~\forall k. \nonumber
\eeqn

Such relaxation can be viewed as considering only one knapsack with
larger capacity. A tight upper bound can be obtained by solving this
relaxed problem, for which the previously proposed dynamic
programming algorithm can be used considering the similarity in
problem structure. Also, based on the algorithm provided
in~\cite{Martello1990}, we can have a low-complexity polynomial-time
(of $O(n^2)$) heuristic algorithm for obtaining an approximation
solution. The details can be found in \textbf{Appendix B}.

For the pricing, since branch-and-bound approach can obtain the
exact solution of the WDP problems, the previous joint use of VCG
pricing and base access price can be applied here which preserves
the incentive compatibility, individual rationality, and allocation
efficiency in each level. Similarly, corresponding pricing schemes
for approximate solution can be applied.

\section{Performance Evaluation}

For numerical analysis, we consider a hexagonal system with $L=7$
cells. We consider that an InP owns $C=100$ subchannels and the BS
of the InP is equipped with $A=200$ antennas. Also, the total
transmit power of the BS is equally divided into $500$ power units.
There are two MVNOs each of which reserves $30$ subchannels, $150$
power units, and $50$ antennas (i.e., $C_m^{res} = 30$, $P_m^{res} =
150$, and $A_m^{res} = 50$), and the leftover resources are
available for auction in the upper level. Each MVNO has $50$
subscribed users. And in the simulation, each user requests $p$
units of power and $c$ subchannels, where $p$ and $c$ are integer
random variables uniformly distributed in the interval [0, 10], and
[0, 2], respectively. All the results are obtained by averaging over
1000 simulation runs.

For solving the single-seller multiple-buyer hierarchical
combinatorial auction problem, we adopt both dynamic programming
algorithm and greedy algorithm. For convenience, we use the terms
`DPA' and `GA' to represent the use of dynamic programming-based
algorithm and the use of greedy algorithm in both levels,
respectively.

For comparison purpose, we consider a fixed sharing scheme, where
each MVNO reserves half of the resources. This fixed sharing can
also be viewed as the case where there is no wireless virtualization
for resource sharing (e.g., each MVNO has its own infrastructure and
fixed resources). Accordingly, the hierarchical auction degenerates
to single lower-level combinatorial auctions hold by each MVNO.
Also, we consider a general sharing scheme as the benchmark, where
the MVNOs are not involved and the InP directly holds a single-level
combinatorial auction for resource allocation. We will use the terms
`FS' and `GS' to represent the results obtained by fixed sharing and
general sharing, respectively. We also solve the multiple-seller
multiple-buyer problem and compare the results with those obtained
for the single-seller model. For fair comparison, we consider the
same setting as that in the single-seller model while allowing the
users to access the services from one of the MVNOs. We will use the
term `MS1' and 'MS2' to represent the exact and approximate
solutions obtained for the multiple-seller problem, respectively.

For numerical analysis, we mainly consider three performance metrics
for resource allocation: average social welfare (i.e., the sum value
of all accepted bids), average resource utilization (i.e., the
proportion of resources utilized), and average user satisfaction
(i.e., the ratio of users whose resource requests are satisfied).

We first consider the average social welfare achieved by different
algorithms as shown in Fig.~\ref{fig:SocialEfficiency}. It is
straightforward that general sharing (GS) provides the largest
social welfare which is used as the benchmark. For dynamic
programming algorithm, we consider two group sizes (i.e., the
resources is auctioned in a grouped manner). There are tradeoffs in
selecting the group size in terms of performance and complexity. We
can see that the DP-Algorithm with group size 1 (i.e., DPA1)
outperforms that with group size 5 (i.e., DPA2). We can also observe
that the social welfare obtained by DPA1 is less than that obtained
by general sharing. This indicates that although the social
efficiency can be achieved for each level, there could exist a gap
between the social welfare obtained by the entire hierarchical
auction with the global optimal one obtained by general sharing.
This gap represents the tradeoff between global social efficiency
and the flexibility of intra-slice customization. Also, we can see
that the greedy algorithm provides good solutions
compared with the dynamic programming-based algorithm. We can also
observe that the welfare obtained for multiple-seller setting is
larger than that for single-seller problem. This is due to
flexibility introduced by allowing dynamic user association. We
can observe that all the proposed dynamic resource sharing schemes
(i.e., DPA, GA, and MS) outperform the fixed sharing scheme. This
 indicates the resource utilization gain by having wireless
virtualization.

\begin{figure}[th]
    \begin{center}
    \includegraphics[width=0.38\textwidth]{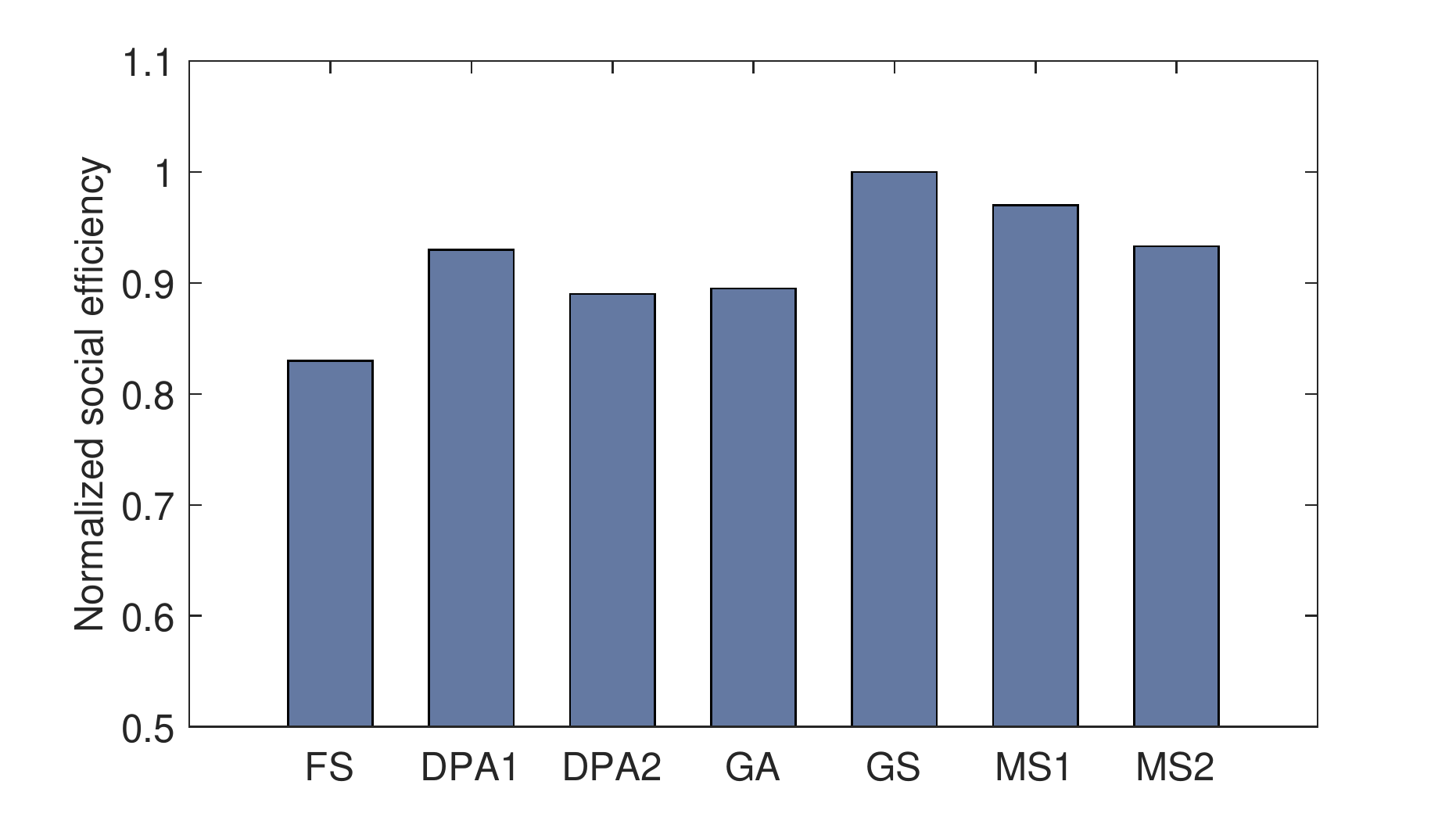}
    \end{center}
    \caption{Normalized average social efficiency achieved by different algorithms. ``FS" = Fixed sharing, ``DPA1" = DP-Algorithm with group size 1, ``DPA2" = DP-Algorithm with group size 5, ``GA" = Greedy algorithm, ``GS" = General sharing, ``MS1" = Multiple seller exact solution, ``MS2" = Multiple seller approximate solution.}
    \label{fig:SocialEfficiency}
\end{figure}

We then investigate the average resource utilization achieved by
different algorithms as shown in Fig.~\ref{fig:AveUtilization}. Here
we use the subchannel utilization as an example. The comparison
results are similar to those for the social welfare.

\begin{figure}[th]
    \begin{center}
    \includegraphics[width=0.38\textwidth]{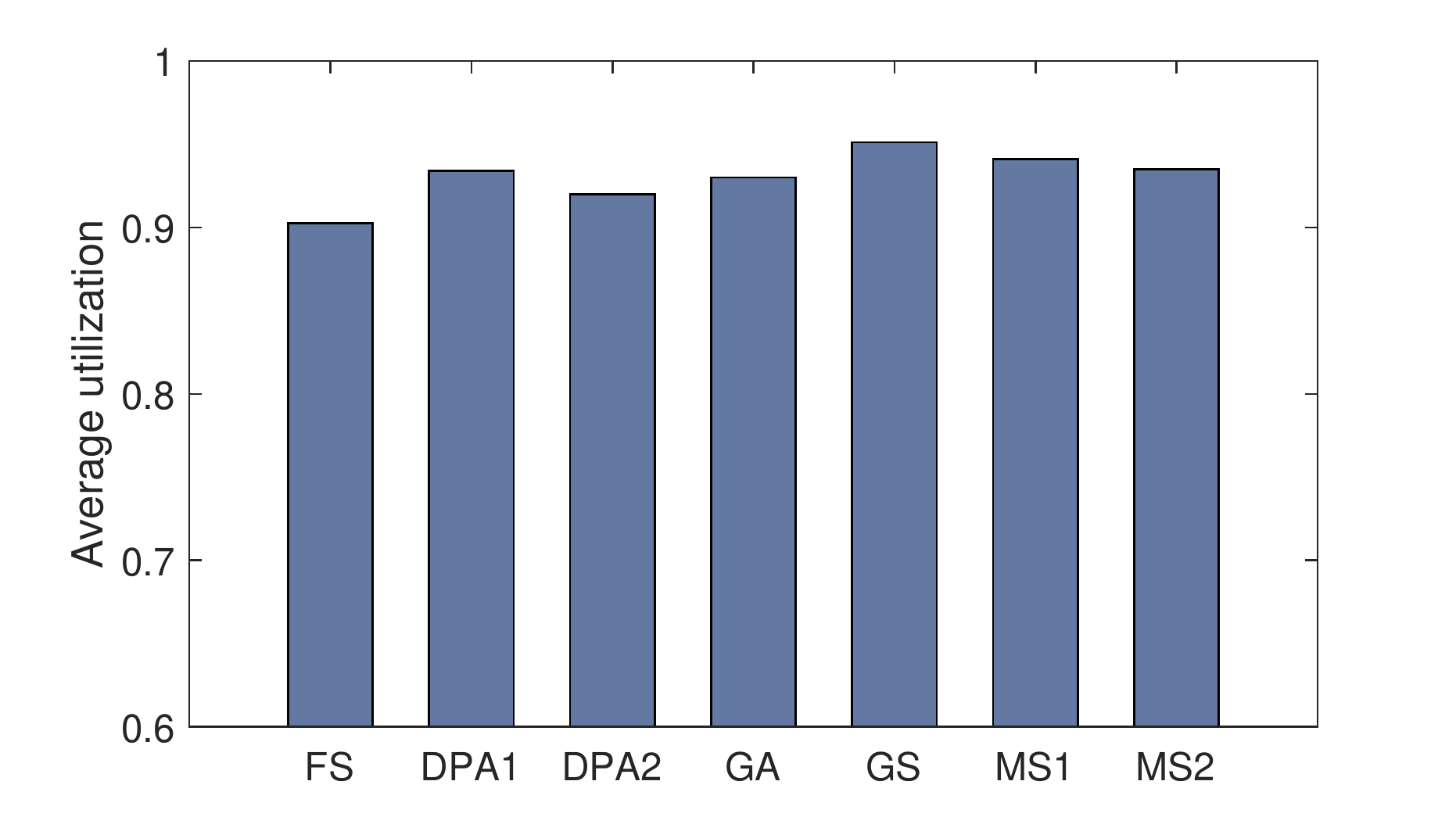}
    \end{center}
    \caption{Average subchannel utilization achieved by different algorithms.}
    \label{fig:AveUtilization}
\end{figure}

\begin{figure}[th]
    \begin{center}
    \includegraphics[width=0.33\textwidth]{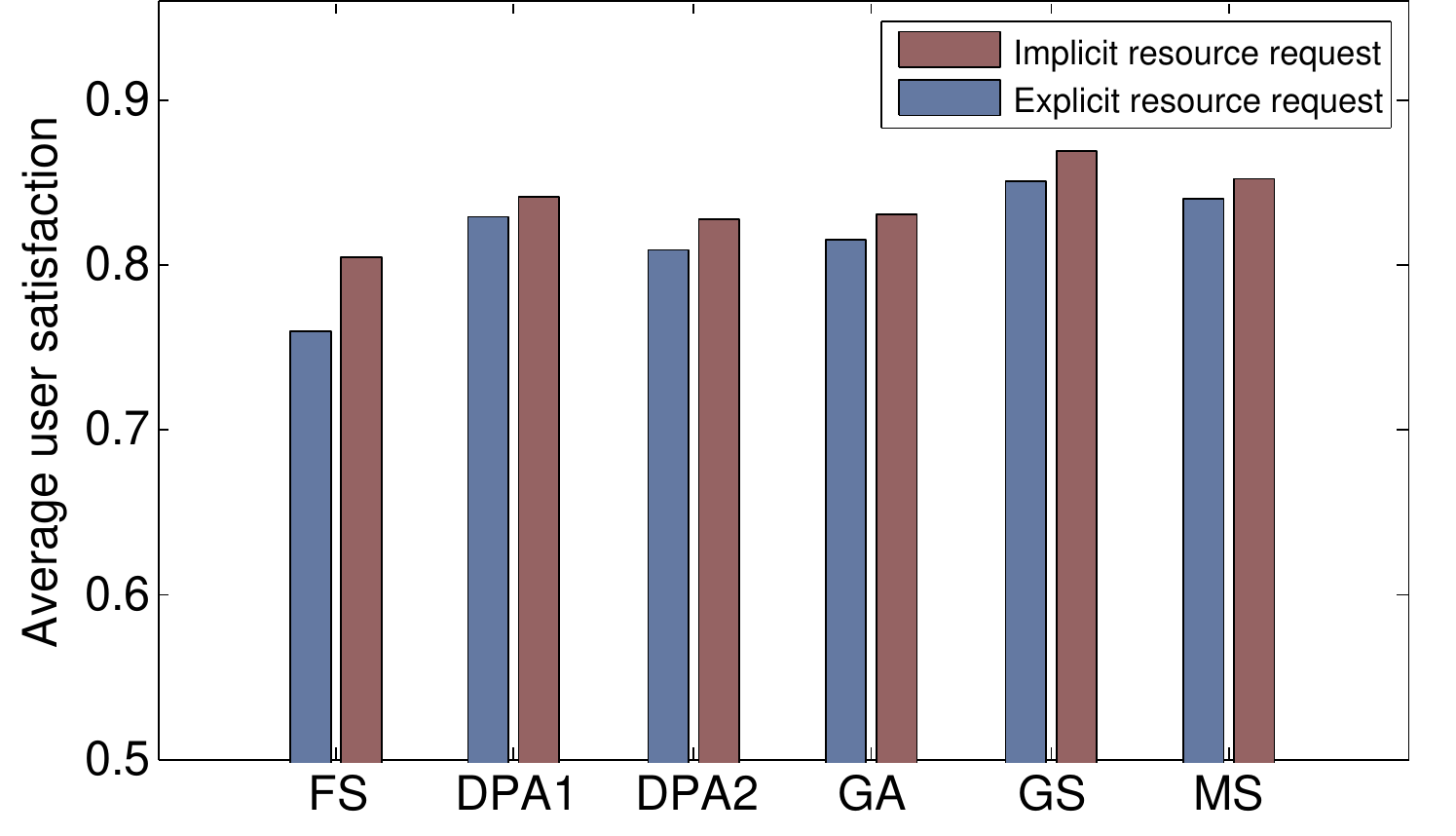}
    \end{center}
    \caption{Average user satisfaction achieved by different algorithms.}
    \label{fig:UserSatisfaction}
\end{figure}

We also compare the average user satisfaction obtained by different
algorithms and compare the results for explicit resource request
case and implicit resource request as shown in
Fig.~\ref{fig:UserSatisfaction}. We can see that the implicit
resource request model can achieve better user satisfaction. This is
due to the benefits from general valuation compared with
single-minded model. Specifically, for implicit resource request,
these exist multiple resource allocation strategies to achieve the
target rate, and if anyone of these strategies is accepted, the user
request is satisfied.

\begin{figure}[th]
    \begin{center}
    \includegraphics[width=0.35\textwidth]{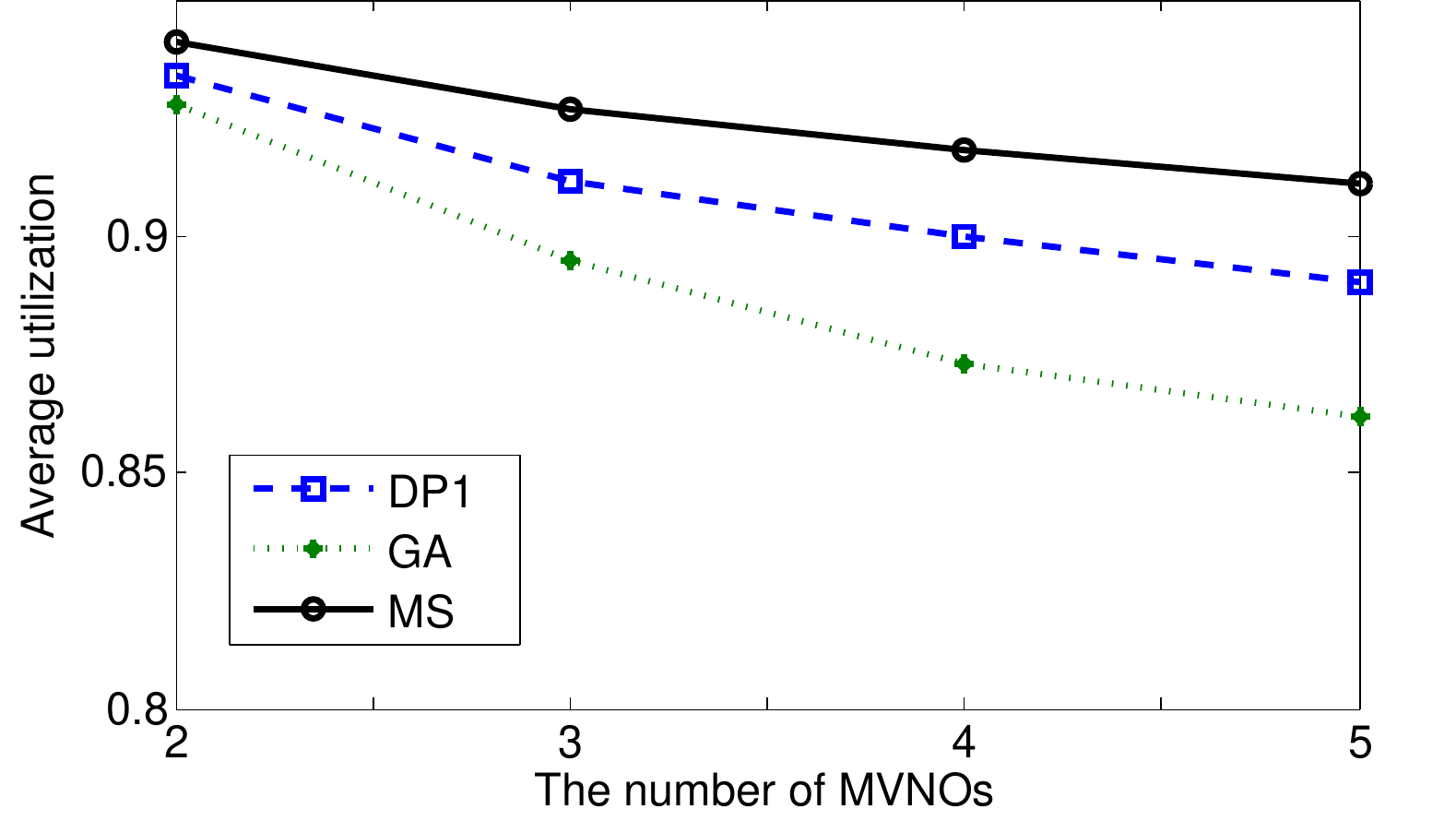}
    \end{center}
    \caption{Average resource utilization achieved by different algorithms with different number of MVNOs.}
    \label{fig:UtilizationVsMVNO}
\end{figure}

The impact of different number of MVNOs on the average resource
utilization is investigated in Fig.~\ref{fig:UtilizationVsMVNO}.
Note that the total number of resources of the InP and the total
number of users are fixed for different number of MVNOs. We can
observe that the utilization decreases with the increased number of
MVNOs. This is due to the fact that increasing the number MVNOs will
decrease the number of resources allocated as well as the number of
users associated to each MVNO, and the statistical multiplexing gain
will decrease.

\section{Conclusion}

We have proposed a hierarchical combinatorial auction mechanism to
jointly address the two-level resource allocation problem for
 virtualization of massive MIMO-based 5G cellular networks.
Specifically, we have considered both single-seller and
multiple-seller hierarchical models, and the corresponding winner
determination problems in two levels have been formulated and
studied. Algorithms have been proposed to solve the WDPs in a
computationally tractable manner and different pricing schemes have
been designed. The properties of the proposed hierarchical auction
mechanism have also been analyzed. The proposed scheme satisfies the
requirements of efficient resource allocation, strict inter-slice
isolation, and the ability of intra-slice isolation. Numerical
results have been presented which show the effectiveness of the
proposed scheme. For the future work, a valuation function taking
into account the fairness and desirability can be designed.

\section*{Appendix A: Summary of main differences of proposed scheme with
existing work}

The main differences between existing work and the proposed scheme
is shown in Table 1.

\begin{table*}
\caption{A summary of main differences of proposed scheme with
existing work} \label{table:summary}
\begin{tabular}{|p{1cm}|p{5.0cm}|p{5.0cm}|p{1.5cm}|p{2.7cm}|}
  \hline
     Schemes & Differences of the model & Pricing scheme & Multi-seller case& Dimensions of auctioned items \\
  \hline
  [13] & There is a social planner controlling the resource allocation in all tiers & First price or VCG-Pricing & No & Single dimension \\
  \hline
  [14] & A middleman can submit bids for multiple items but win at most one & Designed pricing with reserve price in the upper-tier, the same price charged to SUs in the lower-tier & No & Single dimension \\
  \hline
  [15] & Multi-layered spectrum trading through an auction in the upper level and a price demand method in the lower
  level & VCG pricing and two variants of uniform pricing & No &
  Single dimension \\
  \hline
  Proposed & Multi-item combinatorial auction at both tiers & VCG pricing or VCG-like pricing with reserved price & Yes & Multiple dimensions \\
  \hline
\end{tabular}
\end{table*}

\section*{Appendix B: Heuristic Algorithm for Solving the Multi-seller Multi-buyer Problem}

We provide a low-complexity (of $O(n^2)$) heuristic algorithm by extending the
algorithm in [35] for single dimensional items to multiple
dimensions. Specifically, in the initialization phase, we re-index
all users and all MVNOs. Then an initial feasible solution can be
obtained by applying the greedy algorithm. After obtaining the
initial solution, the algorithm improves the solution through local
exchanges. The details of the algorithm are described as follows.

\begin{algorithm}
\caption{Initialization}
\begin{algorithmic}
\STATE 1. Set $|S_k| = \omega_c c_k + \omega_p p_k, ~\forall k$;
\STATE 2. Re-index all users such that
    \beqn
        \frac{b_1(S_1)}{\sqrt{|S_1|}} \geq
        \frac{b_2(S_2)}{\sqrt{|S_2|}} \geq \cdots \geq
        \frac{b_{K}(S_{K})}{\sqrt{|S_{K}|}};
    \eeqn
\STATE 3. Set $\bar{C}_m = \omega_c \hat{C}_m + \omega_p \hat{P}_k,
~\forall m$; \STATE 4. Re-index all MVNOs such that
    \beqn
        \bar{C}_1 \leq \bar{C}_2 \leq \cdots \leq \bar{C}_M.
    \eeqn
\end{algorithmic}
\end{algorithm}

\begin{algorithm}
\caption{Greedy Algorithm}
\begin{algorithmic}
\STATE Denote by $x_k$ the assignment index. Specifically, $x_k=0$
if user $k$ is currently unassigned, and $x_k$ equals the index of
the MVNO it is assigned to, otherwise. \STATE Initialization: set
$x_k = 0, \forall k$; $\tilde{C}_m = \hat{C}_m, \forall m$ and
$\tilde{P}_m = \hat{P}_m, \forall m;$ \FOR{$k := 1$ to $K$}
    \IF {$x_k=0$, $p_k \leq \tilde{P}_m$, and $c_k \leq \tilde{C}_m$}
        \STATE $x_k := m;$
        \STATE $\tilde{C_m} = \tilde{C_m} - c_k;$ ~~ $\tilde{P_m} = \tilde{P_m} - p_k;$
    \ENDIF
\ENDFOR
\end{algorithmic}
\end{algorithm}

\begin{algorithm}
\caption{Initial Solution}
\begin{algorithmic}
\STATE $z:=0;$ \FOR{$k:=1$ to $K$}
    \STATE $x_k :=0;$
\ENDFOR \FOR{$m:=1$ to $M$}
    \STATE Call Greedy Algorithm
\ENDFOR
\end{algorithmic}
\end{algorithm}

\begin{algorithm}
\caption{Rearrangement}
\begin{algorithmic}
\STATE $z:=0;$~$m:=1$;
\FOR{$k:=K$ to 1}
    \IF{$x_k >0$}
        \STATE let $l$ be the first index in $\{m,\ldots,M\}\bigcup \{1,\ldots,m-1\}$ such that $c_k\leq \tilde{C}_l$ and $p_k\leq \tilde{P}_l$;
        \IF{no such $l$}
            \STATE $x_k:=0$
        \ELSE
            \STATE $x_k:= l;$ ~$\tilde{C}_l = \tilde{C}_l - c_k;$ ~$\tilde{P}_l = \tilde{P}_l - c_k;$ ~ $z:=z+b_k;$
            \IF{$l<m$}
                \STATE $m:=l+1;$
            \ELSE
                \STATE $m:=1;$
            \ENDIF
        \ENDIF
    \ENDIF
\ENDFOR \FOR{$m:=1$ to $M$}
    \STATE Call Greedy Algorithm
\ENDFOR
\end{algorithmic}
\end{algorithm}

\begin{algorithm}
\caption{First improvement}
\begin{algorithmic}
\FOR{$k:=1$ to $K$}
    \IF{$x_k>0$}
        \FOR{$j:=k+1$ to $K$}
            \IF {$0<x_j \neq x_k$}
                \STATE $h:= \arg \max\{c_j,c_k\};$ ~ $l:= \arg \min\{c_j,c_k\};$
                \STATE $d_c:= c_h - c_l;$ ~ $d_p:= p_h - p_l;$
                \IF{$d_c\leq \tilde{C}_{x_l}$ and $\tilde{C}_{x_h} \geq \min\{w_u:x_u=0\}$}
                    \STATE $t:=\arg \max\{b_u:x_u=0 and w_u\leq \tilde{C}_{x_h}+d\};$
                    \STATE $\tilde{C}_{x_h} := \tilde{C}_{x_h}+d_c-c_t;$ ~ $\tilde{C}_{x_l} := \tilde{C}_{x_t}-d_c;$
                    \STATE $\tilde{P}_{x_h} := \tilde{P}_{x_h}+d_p-p_t;$ ~ $\tilde{P}_{x_l} := \tilde{P}_{x_t}-d_p;$
                    \STATE $x_t:=x_h;$ ~$x_h:=x_l;$ ~$x_l:=x_t;$ ~$z:=z+b_t;$
                \ENDIF
            \ENDIF
        \ENDFOR
    \ENDIF
\ENDFOR
\end{algorithmic}
\end{algorithm}

\begin{algorithm}
\caption{Second improvement}
\begin{algorithmic}
\FOR{$k:=K$ to 1}
    \IF{$x_k>0$}
        \STATE $\tilde{C}:=\tilde{C}_{x_k} + c_k;$~$\tilde{P}:=\tilde{P}_{x_k} + p_k;$~$Y:= \emptyset;$
        \FOR{$j:=1$ to $K$}
            \IF {$x_k=0$, $c_k \leq \tilde{C}$, and $p_k \leq \tilde{P}$}
                \STATE $Y:=Y \bigcup {j};$ ~ $\tilde{C}:= \tilde{C}-c_k;$ ~$\tilde{P}:= \tilde{P}-p_k;$
            \ENDIF
            \IF {$\sum_{j\in Y}b_k > b_k$}
                \FOR{each $k\in Y$}
                    \STATE $x_j:=x_k;$ ~ $\tilde{C}_{x_k}:=\tilde{C};$ ~ $\tilde{P}_{x_k}:=\tilde{P};$ ~$x_k:=0;$ ~$Z:=Z+\sum_{j\in Y}b_j - b_k;$
                \ENDFOR
            \ENDIF
        \ENDFOR
    \ENDIF
\ENDFOR
\end{algorithmic}
\end{algorithm}

\begin{IEEEbiography}
[{\includegraphics[width=1in,height=1.25in,clip,keepaspectratio]{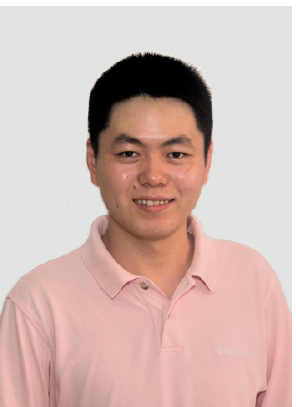}}]
{Kun Zhu}  received his Ph.D. degree in 2012 from School of Computer
Engineering, Nanyang Technological University, Singapore. He is
currently with College of Computer Science and Technology at Nanjing
University of Aeronautics and Astronautics. He was a Post-Doctoral
Research Fellow with the Wireless Communications, Networks, and
Services Research Group in the Department of Electrical and Computer
Engineering at the University of Manitoba, Canada. His research
interests are in the areas of resource management in multi-tier
heterogeneous wireless networks.

\end{IEEEbiography}

\begin{IEEEbiography}
[{\includegraphics[width=1in,height=1.25in,clip,keepaspectratio]{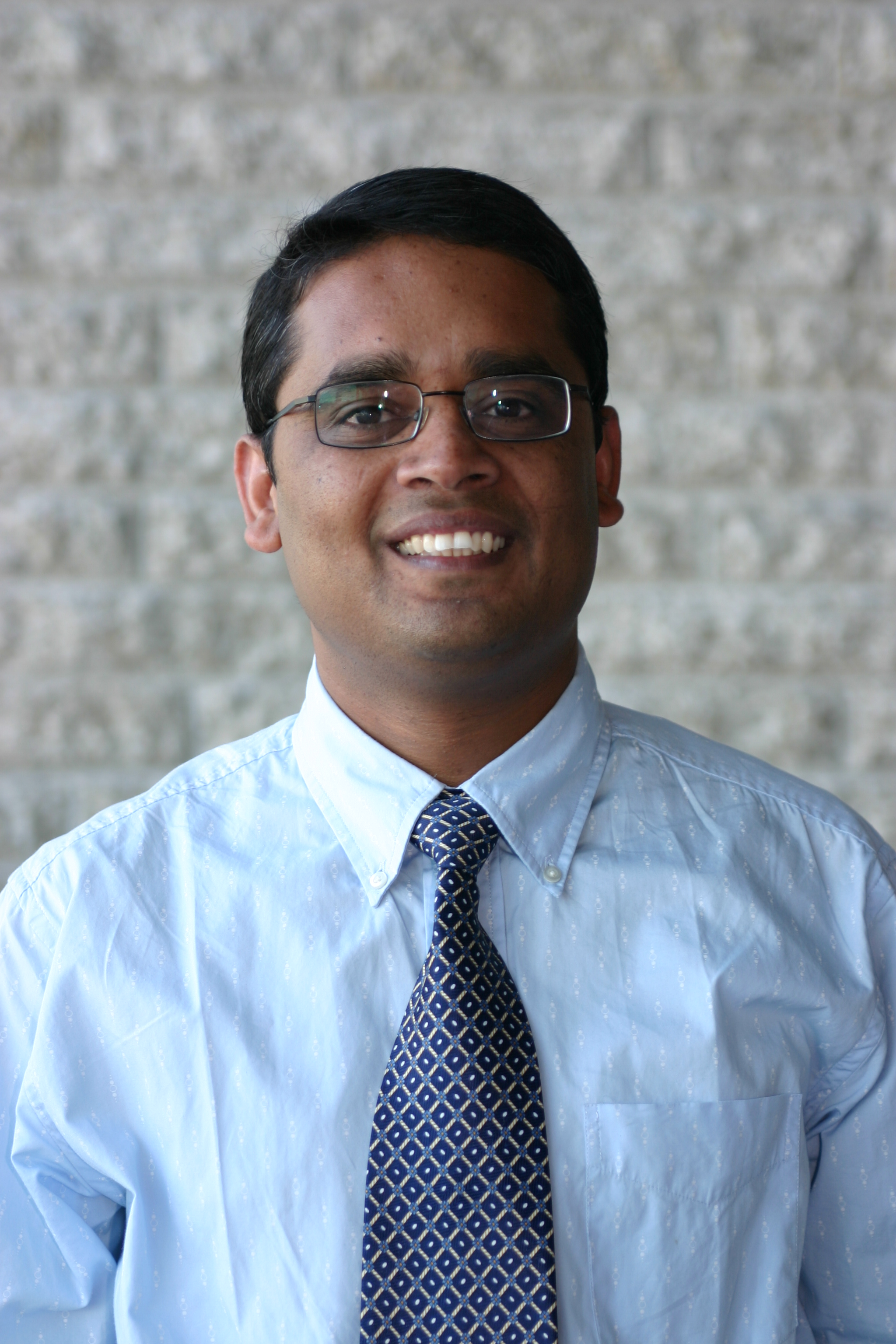}}]
{Ekram Hossain} (F'15)
is a Professor in the Department of Electrical and Computer
Engineering at University of Manitoba, Winnipeg,
Canada. He received his Ph.D. in Electrical
Engineering from University of Victoria,
Canada, in 2001. Dr. Hossain's current research
interests include design, analysis, and optimization
of wireless/mobile communications networks, cognitive
radio systems, and network economics. 
He has authored/edited several books in these areas
(http://home.cc.umanitoba.ca/$\sim$hossaina). 
He was elevated to an IEEE Fellow ``for contributions to spectrum management and resource allocation in cognitive and cellular radio networks". 
Dr. Hossain has won several research awards including
the IEEE Communications Society Transmission, Access, and Optical Systems (TAOS) Technical Committee's Best Paper Award in IEEE Globecom 2015, University of Manitoba Merit Award in 2010 and 2014 (for Research and
Scholarly Activities), the 2011 IEEE Communications Society Fred Ellersick
Prize Paper Award, and the IEEE Wireless Communications and Networking
Conference 2012 (WCNC'12) Best Paper Award. He is a Distinguished Lecturer of the
IEEE Communications Society (2012-2015). He is a registered Professional
Engineer in the province of Manitoba, Canada.
\end{IEEEbiography}

\end{document}